\documentclass[aps,prl,amsmath,amssymb,preprintnumbers,superscriptaddress,twocolumn]{revtex4-1}

\usepackage[utf8]{inputenc}
\usepackage{amssymb}
\usepackage{amsfonts}
\usepackage{amsmath}
\usepackage{amsthm}
\usepackage[pdftex]{graphicx}
\usepackage[usenames,dvipsnames]{color}
\usepackage{braket}
\usepackage[normalem]{ulem}
\usepackage{times}
\usepackage{multirow,bigdelim}
\usepackage{bm}
\usepackage{hyperref}

\newcommand{\sign}{\mathrm{sgn}}
\newcommand{\rank}[1]{\mathrm{rank}{(#1)}}
\newcommand{\diag}{\mathrm{diag}}
\newcommand{\trans}[1]{{#1}^{\ensuremath{\mathsf{T}}}} 
\newcommand{\Tr}{\mathrm{Tr}}
\newcommand{\chG}{\Phi} 
\newcommand{\chH}{\chG^{\rm C}} 
\newcommand{\chPI}{\chG^{\rm TH}} 
\newcommand{\chPS}{\chG^{\rm F}} 
\newcommand{\chCS}{\chG^{\rm CS}} 
\newcommand{\chSQ}{\chG^{\rm SQ}} 
\newcommand{\eff}{\bm{Y}} 
\newcommand{\effH}{\eff_{\rm C}} 
\newcommand{\effPI}{\eff_{\rm TH}} 
\newcommand{\effPS}{\eff_{\rm F}} 
\newcommand{\effCS}{\eff_{\rm CS}} 
\newcommand{\effSQ}{\eff_{\rm SQ}} 
\newcommand{\X}{\bm{X}} 
\newcommand{\XH}{\X_{\rm C}} 
\newcommand{\XPI}{\X_{\rm TH}} 
\newcommand{\XPS}{\X_{\rm F}} 
\newcommand{\XCS}{\X_{\rm CS}} 
\newcommand{\XSQ}{\X_{\rm SQ}} 
\newcommand{\disp}{\bm{\alpha}}
\newcommand{\dispin}{\disp_{\rm in}}
\newcommand{\dispinmod}{\bar{\disp}_{\rm in}}
\newcommand{\dispout}{\disp_{\rm out}}
\newcommand{\CM}{\bm{V}}
\newcommand{\CMBE}{\bm{V}_{\rm BE}}
\newcommand{\CMBEmod}{\bm{\bar{V}}_{\rm BE}}

\newcommand{\CMin}{\CM_{\rm in}}
\newcommand{\CMinth}{\CM_{\rm in,th}}
\newcommand{\CMmod}{\CM_{\rm mod}}
\newcommand{\CMinmod}{\bar{\CM}_{\rm in}}
\newcommand{\CMout}{\CM_{\rm out}}
\newcommand{\CMoutmod}{\bar{\CM}_{\rm out}}

\newcommand{\rhoinmod}{\hat{\bar{\rho}}} 
\newcommand{\rhoinmodG}{\hat{\bar{\rho}}^\mathrm{G}} 
\newcommand{\rhoX}{\hat{\rho}_x}
\newcommand{\rhoG}{\hat{\rho}^\mathrm{G}}

\newcommand{\CapG}{C^{\mathrm{G}}}
\newcommand{\CapGchi}{C^{\mathrm{G}}_\chi}
\newcommand{\chiG}{\chi^{\mathrm{G}}}
\newcommand{\muG}{\mu^{\mathrm{G}}}
\newcommand{\U}{U}
\newcommand{\Uone}{{\U}_1}
\newcommand{\Utwo}{{\U}_2}

\newcommand{\SU}{\bm{M}}
\newcommand{\SUtilde}{\tilde{\bm{M}}}
\newcommand{\SUtildetrans}{\trans{\tilde{\bm{M}}}}
\newcommand{\SUchG}{\bm{M}}
\newcommand{\SUchGtrans}{\trans{\bm{M}}}
\newcommand{\SUtrans}{\trans{\bm{M}}}
\newcommand{\SUone}{\bm{M}_1}
\newcommand{\SUonetrans}{\trans{\bm{M}}_1}
\newcommand{\SUtwo}{\bm{M}_2}
\newcommand{\SUtwotrans}{\trans{\bm{M}}_2}
\newcommand{\Sgenrot}{\bm{O}}
\newcommand{\Srot}{\bm{\Theta}}
\newcommand{\SrotchG}{\bm{\Theta}}

\newcommand{\SrotPS}{\bm{\Theta}_{\rm F}}
\newcommand{\SrotPStrans}{\trans{\bm{\Theta}}_{\rm F}}
\newcommand{\Ssq}{\bm{S}}
\newcommand{\Scnot}{\bm{S}_{\rm CNOT}}

\newcommand{\SrotXone}{{{\Srot}_1}_X}
\newcommand{\SrotXonetrans}{{\trans{\Srot}_1}_X}
\newcommand{\SrotXtwo}{{{\Srot}_2}_X}

\newcommand{\thetaXone}{{\theta_1}_X}
\newcommand{\thetaXtwo}{{\theta_2}_X}

\newcommand{\Nthr}{\bar{N}_{\rm thr}}
\newcommand{\Ninmod}{\bar{N}}

\newcommand{\seout}{\nu}
\newcommand{\seoutmod}{\bar{\nu}}
\newcommand{\sey}{y}
\newcommand{\seythr}{y_{\rm thr}}

\newcommand{\EG}{E^{G}}
\newcommand{\setN}{\mathcal E}
\newcommand{\setMu}{\mathcal F}

\newcommand{\thickhline}{%
    \noalign {\ifnum 0=`}\fi \hrule height 1pt
    \futurelet \reserved@a \@xhline
}

\newcounter{thm} 
\newcounter{thmSup} 
\newcounter{cor} 
\newcounter{corSup} 
\newtheorem{thmmapeq}[thm]{Theorem}
\newtheorem{thmmapeqSup}[thmSup]{Theorem}
\newtheorem{corCap}[cor]{Corollary}
\newtheorem{corCapG}[cor]{Corollary}
\newtheorem{corCapBound}[cor]{Corollary}

\newtheorem{corCapGSup}[corSup]{Corollary}

\def\comment#1{}

\def\section#1{{\par\em #1.---}}
\def\togli#1{}

\begin{document}

\title{Equivalence Relations for the Classical Capacity of Single-Mode Gaussian Quantum channels}

\author{Joachim Sch\"afer}\affiliation{QuIC, Ecole Polytechnique de Bruxelles, CP 165, Universit\'e Libre de Bruxelles, 1050 Brussels, Belgium}
\author{Evgueni Karpov}\affiliation{QuIC, Ecole Polytechnique de Bruxelles, CP 165, Universit\'e Libre de Bruxelles, 1050 Brussels, Belgium}
\author{Ra\'ul Garc\'ia-Patr\'on}\affiliation{QuIC, Ecole Polytechnique de Bruxelles, CP 165, Universit\'e Libre de Bruxelles, 1050 Brussels, Belgium}
\affiliation{Max-Planck-Institut f\"ur Quantenoptik, Hans-Kopfermann-Stra\ss e 1, 85748 Garching, Germany}
\author{Oleg V. Pilyavets}\affiliation{QuIC, Ecole Polytechnique de Bruxelles, CP 165, Universit\'e Libre de Bruxelles, 1050 Brussels, Belgium}
\author{Nicolas J. Cerf}\affiliation{QuIC, Ecole Polytechnique de Bruxelles, CP 165, Universit\'e Libre de Bruxelles, 1050 Brussels, Belgium}


\begin{abstract} 
We prove the equivalence of an arbitrary single-mode Gaussian quantum channel and a newly defined fiducial channel preceded by a phase shift and followed by a Gaussian unitary operation.  
This equivalence implies that the energy-constrained classical capacity of any single-mode Gaussian channel can be calculated based on this fiducial channel, which is furthermore simply realizable with a beam splitter, two identical single-mode squeezers, and a two-mode squeezer. In a large domain of parameters, we also provide an analytical expression for the Gaussian classical capacity, exploiting its additivity, and prove that the classical capacity cannot exceed it by more than $1/\ln 2$ bits.
\end{abstract}

\maketitle
\section{Introduction}Quantum channels play a key role in quantum information theory.
In particular,
bosonic Gaussian channels model most optical communication links, such as
optical fibers or free space information transmission~\cite{GaussRev1,GaussRev2}. One of the central characteristics of quantum channels is their classical capacity.
A lot of attention has already been devoted to the study of the classical capacity of Gaussian channels \mbox{\cite{G04,Channels1,Channels2,Channels6,Channels7,ChannelsSH1,ChannelsSH2,ChannelsSH3,H98,Channels3,Channels4,Channels5,HW01,H05,LPM09,SDKC09,PLM09,SKC10,SKC11,SKC12}}. Since Gaussian encodings are more relevant for experimental implementations, easier to work with analytically, and conjectured to be optimal \cite{HW01}, the so-called Gaussian classical capacity was evaluated for specific Gaussian channels \cite{Channels4,Channels5,H05,LPM09,SDKC09,PLM09,SKC10,SKC11,SKC12}. 

In this Letter, we greatly simplify the calculation of these capacities~\cite{GaussianComment} for an arbitrary single-mode Gaussian channel. Namely, we show that any such channel is indistinguishable from a newly defined fiducial channel, preceded by a phase shift and followed by a general Gaussian unitary. Since neither the phase shift at the channel's input nor the Gaussian unitary at the channel's output affects the input energy constraint or changes the output entropy, the capacities of this channel are equal to those of the fiducial channel. This conclusion also holds for any cascade of Gaussian channels since the latter is equivalent to another Gaussian channel. Our results allow us to go beyond previous works on the Gaussian capacity \cite{PLM09,SKC10,SKC11} and provide its unified analytical expression valid for any Gaussian channel in some energy range, where it is additive. In this range we prove that the capacity cannot exceed the Gaussian capacity by more than $1/\ln 2$ bits (generalizing \cite{KS13}), the latter becoming the actual capacity if the minimum-output entropy conjecture for phase-insensitive Gaussian channels~\cite{HW01,GGLMS04} is true.

\section{Gaussian channel}Let $\rhoG(\disp,\CM)$ be a single-mode Gaussian state, where the coherent vector $\disp \in \mathbb{R}^{2}$ and the covariance matrix (CM) $\CM \in \mathbb{R}^2 \times \mathbb{R}^2$ are the first- and second-order moments of the $2$ dimensionless quadratures, respectively, with $\hbar=1$. Then, a single-mode Gaussian channel $\chG$ is a completely positive trace-preserving map which is closed on the set of Gaussian states~\cite{HW01}. It transforms input states with moments $\{\dispin,\CMin\}$ to output states with moments $\{\dispout,\CMout\}$ according to
\begin{equation}\label{eq:chgauss}
	\dispout  =\X \dispin + \bm{\delta},\quad
	\CMout  =\X \CMin \trans{\X} + \eff,
\end{equation}
where $\bm{\delta}$ is the displacement induced by the channel, $\X$ is a $2 \times 2$ real matrix, and $\eff$ is a $2 \times 2$ real, symmetric, and non-negative matrix. For simplicity, we choose $\bm{\delta}=0$ in what follows (the capacity is not affected by $\bm{\delta}$), and focus on the action of the map $\chG$ on second-order moments using the simplified notation $\chG(\CMin) = \CMout$. Then, the map $\chG$  is fully characterized by matrices $\X$ and $\eff$, which must satisfy  $\eff+i\left(\bm{\Omega}-\X \bm{\Omega} \trans{\X}\right)/2 \ge 0$ \cite{HW01}, where
\begin{equation}
	\bm{\Omega}=\begin{pmatrix} 0 & 1\\-1 & 0 \end{pmatrix},
\end{equation}
is the symplectic form~\cite{DeGosson}. In the following, we use the parameters 
\begin{equation}
	\tau = \det{\X}, \quad y = \sqrt{\det{\eff}},    \label{def-tau-y}
\end{equation}
where $\tau$ may be a channel transmissivity (if $0 \le \tau\le1$) or amplification gain (if $\tau\ge1$), while $y$ characterizes the added noise. The map $\chG$ describes a
quantum channel if $\sey \ge |\tau-1|/2$~\cite{GNLSC12}. Moreover, it is an entanglement breaking channel if $\sey~\ge~(|\tau|+1)/2$ \cite{H08}. The single-mode Gaussian channels can therefore be conveniently represented in a ($\tau$,$\sey$) plane, see Fig.~\ref{fig:entbreak}.
\section{Canonical decomposition}Any single-mode Gaussian channel $\chG$ can be decomposed as $\chG = \Utwo
\circ \chH \circ \Uone$, where $\Uone$ and $\Utwo$ are Gaussian unitaries, and $\chH$ is a canonical
channel characterized by the matrices $(\XH,\effH)$~\cite{H07CGH061,H07CGH062,H07CGH063}. The action of a
Gaussian unitary $U$ on a Gaussian state can be completely specified by a symplectic transformation $\SU$
acting on the second-order moments of the state (we ignore first-order moments), so that the canonical
decomposition may be written as $(\Utwo \circ \chH \circ \Uone)(\CMin)=\SU_2 \, \chH(\SU_1 \CMin
\SUtrans_1) \, \SUtrans_2$. One can define seven classes of canonical channels $\chH$ (see
Table~\ref{table:channels}) \cite{H07CGH061,H07CGH062,H07CGH063}.
The first five channels in Table~\ref{table:channels} can be treated together, and we refer to them collectively as thermal channels, $\chPI_{(\tau,\sey)}$:
\begin{equation}\label{eq:chpi}
	\XPI = \begin{pmatrix}
		\sqrt{|\tau|} & 0\\0 & \sign(\tau) \sqrt{|\tau|} 
	\end{pmatrix}, \;
	\effPI = \begin{pmatrix}
		\sey & 0\\0 & \sey
	\end{pmatrix},
\end{equation}
where $\sign(\tau)=-1$ if $\tau < 0$ and $\sign(\tau)=1$ if $\tau \ge 0$. 
As shown in Fig.~\ref{fig:graph_ad} (a), any channel $\chPI$ can be physically realized by a beam splitter with transmissivity $T$ followed by a two-mode squeezer (TMS) with gain $G$ \cite{GNLSC12}. For the zero-transmission ($\tau=0$), lossy ($0 \le \tau \le 1$), amplification ($\tau \ge 1$), and classical additive-noise channel ($\tau=1$), the output is given by the signal's output of the TMS, and these four canonical channels correspond to phase-insensitive channels.
For the fifth canonical channel, i.e., the phase-conjugating channel ($\tau<0$), the output is given by the idler's output of the TMS.
These five channels map any thermal state to a thermal state, so we call them thermal channels. 
Each particular channel $\chPI_{(\tau,\sey)}$ corresponds to a single point in Fig.~\ref{fig:entbreak}, where the relations between $(\tau,\sey)$ and $(T,G)$ are given in Table \ref{table:channels}. Finally, the sixth and seventh canonical channels are the classical signal (or quadrature erasing) channel $\chCS$ and the single-quadrature classical noise channel $\chSQ$, which are not thermal channels (see \cite{SUPPMAT}).

\begin{figure}[t]
	\centering
	\includegraphics[width=0.495\textwidth]{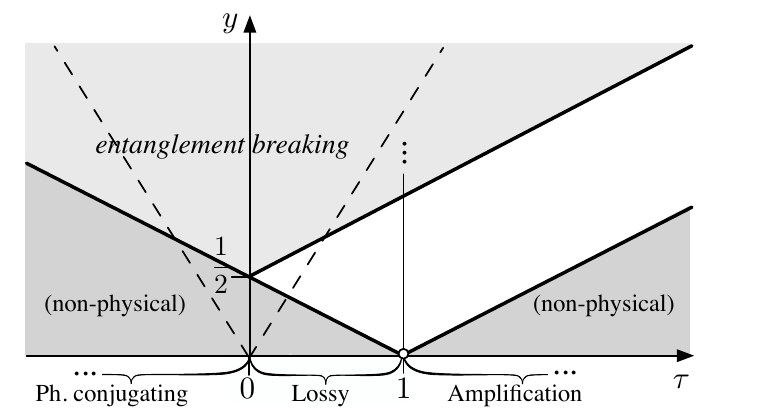}
	\caption{Admissible regions in the parameter space $(\tau,y)$ for Gaussian quantum channels. Each thermal channel $\chPI_{(\tau,\sey)}$ is associated with a particular point $(\tau,y)$. 
	The vertical line $\tau=0$ corresponds to the zero-transmission channel as well as the classical signal channel $\chCS$. The vertical line $\tau=1$ corresponds to the classical additive-noise channel. Both the perfect transmission channel and the single-quadrature classical noise channel $\chSQ$ correspond to $(\tau=1,y=0)$. The Gaussian capacity of  $\chPS_{(\tau,\sey,s)}$ is additive if $\Ninmod \ge \Nthr$. This is equivalent to $\sey \le \seythr = |\tau|(e^{-2|s|}(1+2\Ninmod)-1)/(1-e^{-4|s|})$. An example of $\seythr$ is given by the dashed line, where $\bar{N}=0.5$ and $s=0.12$.}
	\label{fig:entbreak}
\end{figure}

\section{Fiducial channel}
Now, our central point is that the above canonical decomposition is not always useful for evaluating capacities of bosonic channels with input energy constraint (which is needed, otherwise the capacities are infinite). 
Indeed, the Gaussian unitary $U_1$ that precedes the canonical channel $\chH$ affects, in general, the input energy.
Therefore, we introduce a new decomposition in terms of a fiducial channel $\chPS$,
where the preceding unitary is passive and does not affect the input energy restriction. 
We show that this decomposition has the major advantage that the energy-constrained capacity of any Gaussian channel reduces to that of the fiducial channel $\chPS$.
The latter generalizes $\chPI$ by introducing squeezing in the added noise
\begin{equation}\label{eq:chPS}
	\XPS = \XPI, \quad
	\effPS = \sey\begin{pmatrix}
		e^{2s} & 0\\0 & e^{-2s}
	\end{pmatrix}.
\end{equation}
Thus, it depends on three parameters $(\tau,\sey,s)$, and we denote it by $\chPS_{(\tau,\sey,s)}$. 
This channel can be physically realized by the setup depicted in Fig.~\ref{fig:graph_ad} (b), where the ``idler'' corresponds again to the output of the phase-conjugating channel and the ``signal'' to that of the other channels.
In the case $0 \le \tau \le 1$, this channel corresponds to the mixing of the input state with an arbitrary squeezed thermal state on a beam splitter with transmissivity $\tau$. 
The fiducial channel $\chPS$ can be used to decompose any Gaussian channel $\chG$ (by taking proper limits, if necessary) \cite{SUPPMAT}.

\begin{figure}[t]
	\centering
	\includegraphics[width=0.49\textwidth]{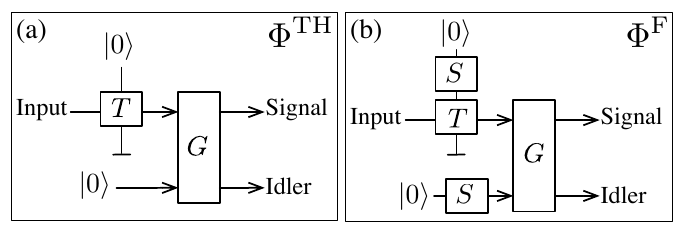}
	\caption{Realization of (a)  the thermal channel $\chPI$  and  (b)  the fiducial channel $\chPS$ by a beam splitter with transmissivity $T$, a two-mode squeezer with gain $G$, and  a single-mode squeezer $S$. Here  $\ket{0}$ stands for the vacuum state, and ``$\dashv$'' denotes ``tracing out'' the mode.}
	\label{fig:graph_ad}
\end{figure}

\begin{thmmapeq}\label{thmmapeq} 
	For a single-mode Gaussian channel $\chG$ defined by matrices $\X$ and $\eff$ with $\tau\ne 0$ and $\sey >0$, there exists a fiducial channel $\chPS_{(\tau,\sey,s)}$ defined by matrices $\XPS(\tau)$, $\effPS(y,s)$ with $\tau$ and $y$ obtained from Eq. (\ref{def-tau-y}), a symplectic transformation $\SUchG$, and a rotation in phase space $\SrotchG$ such that
\begin{equation}	
	\X = \SUchG \, \XPS(\tau) \, \SrotchG, \; \eff = \SUchG \, \effPS(y,s) \, \SUchGtrans,     \label{eq:mapeqPI}
\end{equation}	
	where the explicit dependencies of $\SUchG$, $\SrotchG$, and $s$ on the parameters of the channel $\chG$ are presented in the Supplemental Material \cite{SUPPMAT}. 
\begin{proof}
We only sketch the proof here (see \cite{SUPPMAT} for the full proof). 
First, one finds matrices $\Srot_Y$ and $\Ssq_Y$ such that $\Ssq^{-1}_Y\trans{\Srot}_Y \eff\Srot_Y \Ssq_Y^{-1}=\diag{(y,y)}$, where $\Srot_Y$ and $\Ssq_Y$ denote matrices corresponding to a rotation and a squeezing operation, respectively.
Second, one obtains the singular value decomposition $\X=\SrotXone \Ssq_X \XPS \SrotXtwo$, where $\XPS$ reads as in Eq.~\eqref{eq:chPS}. Then one defines $\SUchG=\SrotXone \Ssq_X \SrotPStrans$, where $\SrotPS$ is found such that $\SUchGtrans \eff \SUchG = \effPS = y \, \diag{(e^{2s},e^{-2s})}$.
The squeezing parameter $s$ depends on all angles and squeezing operations $\Ssq_X, \Ssq_Y$. Finally, one introduces $\SrotPS$ in $\X$, i.e. $\X=\SrotXone \Ssq_X \SrotPStrans \SrotPS \XPS \SrotXtwo = \SUchG \XPS \SrotchG$, where $\SrotchG$ depends on $\SrotXtwo, \SrotPS$, and the sign of $\tau$. Despite Theorem~\ref{thmmapeq} requires that $\tau\ne 0$ and $\sey >0$, or, equivalently, that $\rank{\X}=\rank{\eff}=2$, it can be extended to lower-rank cases with minor modifications \cite{SUPPMAT}. 
\end{proof}
\end{thmmapeq} 

\begin{table*}
	\begin{tabular}{| l || c | c | c | c | c | c | c | c | c |}	
			\hline
				Channel 							& Symbol 				& Class 					& $\XH$ 						& $\effH$ 							& $\tau$		& Domain of $\tau$	& Domain of $y$ 						\\ \noalign{\hrule height 1.5pt}
				Zero-transmission 					& ${\mathcal A}_1$ 	& \multirow{6}{*}{$\chPI$}	& $0$ 						& $(G-1/2) \openone$ 				& $0$			& $0$ 			& $[1/2,\infty)$ 					\\ \cline{1-2} \cline{4-8}			
				Classical additive noise 			& ${\mathcal B}_2$ 	& 							& $\openone$				& $(G-1) \openone$ 				& $TG=1$			& $1$ 			& $[0,\infty)$ 				\\ \cline{1-2} \cline{4-8}
				Lossy 								& ${\mathcal C}_L$ 	& 							& $\sqrt{\tau} \openone$	& $[G(1-T/2)-1/2]\openone$ 		& $TG$			& $[0,1]$ 		& $[(1-\tau)/2,\infty)$ 		\\ \cline{1-2} \cline{4-8}
				Amplification 						& ${\mathcal C}_A$ 	& 							& $\sqrt{\tau} \openone$	& $[G(1-T/2)-1/2]\openone$ 		& $TG$			& $[1,\infty)$	& $[(\tau-1)/2,\infty)$		\\ \cline{1-2} \cline{4-8}
				Phase-conjugating 					& ${\mathcal D}$ 	& 							& $\sqrt{|\tau|} \sigma_z$	& $[(1-T)(G-1)/2+G/2]\openone$ 		& $-T(G-1)$		& $(-\infty,0]$ & $[(1-\tau)/2,\infty)$ 	\\ \noalign{\hrule height 1.5pt}
				Classical-signal					& ${\mathcal A}_2$ 	& $\chCS$ 					& $(\openone+\sigma_z)/2$	& $(G-1/2) \openone$ 				& 0				& 0 			& $[1/2,\infty)$ 						\\ \hline
				Single-quad. cl. noise 				& ${\mathcal B}_1$ 	& $\chSQ$					& $\openone$				& $(\openone-\sigma_z)/4$			& 1				& 1				& 0 						\\ \hline
			\end{tabular}
	\caption{Canonical channels $\chH$ as defined in \cite{H07CGH061,H07CGH062,H07CGH063}, and their new representation in terms of $\chPI, \chCS$, $\chSQ$ and the corresponding matrices $(\XH,\effH)$, where $\sigma_z=\diag(1,-1)$. The  transmissivity  $T \in [0,1]$ of the beam splitter and the gain  $G \ge 1$ of the two-mode squeezer correspond to the physical schemes in Fig.~\ref{fig:graph_ad} and \cite{SUPPMAT}.}
	\label{table:channels}
\end{table*}

\section{Capacities}The energy-constrained capacity $C$ of the Gaussian channel $\chG$ is defined as the maximal amount of bits that can be transmitted per use of the channel $\chG$ given the mean photon number $\Ninmod$ at its input, i.e. \cite{SW97,H98}
\begin{equation}
	C(\chG,\Ninmod) = \lim_{n \to \infty}\frac{1}{n}C_\chi(\chG^{\otimes n},n\Ninmod),
	\label{eq:capacity}
\end{equation}
where $n$ is the number of channel uses, and $C_\chi$ is the one-shot capacity of the channel, i.e. 
\begin{equation}\label{eq:chi}
	\begin{split}
		C_\chi(\chG,\Ninmod) & = \max_{\mu \, : \, \rhoinmod \in \setN_{\Ninmod}}\chi(\chG,\mu),\\
		\chi(\chG,\mu) & = S(\chG[\rhoinmod]) - \int{\mu(dx) \, S(\chG[\rhoX])}.
	\end{split}	
\end{equation}
Here $S(\hat{\rho}) = -\mathrm{Tr}(\hat{\rho}\log_2{\hat{\rho}})$ is the von Neumann entropy. 
The maximum is taken over all probability measures $\mu(x)$ in the whole space ${\mathcal H}$ of pure symbol states $\rhoX$ such that the average state $\rhoinmod=\int{\mu(dx) \rhoX}$ belongs to the set $\setN_{\Ninmod}$ of states which have a mean photon number not greater than $\Ninmod$. 
Since, in general, the one-shot capacity is not additive \cite{H09}, one has to take the limit in \eqref{eq:capacity}, unless additivity is explicitly proven for the given channel. The decomposition stated in Theorem~\ref{thmmapeq} implies:
\begin{corCap}\label{corCap} 
  	For a single-mode Gaussian channel $\chG$ with parameters $(\tau \ne 0, \sey >0)$, there exists a fiducial channel $\chPS$ as defined in Theorem 1, such that
	\begin{equation}\label{eq:capequal}		
		 C(\chG,\Ninmod) = C(\chPS,\Ninmod).
	\end{equation}
	\begin{proof}
		The symplectic transformation $\SUchG$ that follows $\chPS$ in Theorem~\ref{thmmapeq} does not change the entropies in $\chi$ and there is no energy constraint on the output of the channel.
		Hence, $\SUchG$ can be omitted. Furthermore, the rotation $\SrotchG$ preceding $\chPS$ in Theorem~\ref{thmmapeq} may be regarded as a change of the reference phase that can be chosen arbitrarily; therefore, $\SrotchG$ can be omitted as well. 
		Thus, $C_\chi(\chG,\Ninmod)=C_\chi(\chPS,\Ninmod)$ holds. 
In order to evaluate the one-shot capacity of $\chG^{\otimes n}$ we apply the same reasoning, where the preceding and following transformations are given by $\oplus_{i=1}^n\SUchG$ and $\oplus_{i=1}^n \SrotchG$, respectively. Hence, it follows that $C_\chi({\chG}^{\otimes n},n\Ninmod)=C_\chi(({\chPS})^{\otimes n},n\Ninmod)$ which together with Eq.~\eqref{eq:capacity} implies Eq.~\eqref{eq:capequal}. Note that despite Eq.~\eqref{eq:capequal} requires $\tau\ne 0$ and $\sey >0$, it can be easily extended to the general case \cite{SUPPMAT}. 
	\end{proof}
\end{corCap}

We remark that if the corresponding fiducial channel 
$\chPS_{(\tau,\sey,s)}$ is entanglement breaking, then the one-shot capacities of both $\chPS_{(\tau,y,s)}$ and $\chG$ are additive
\cite{S02HS041,S02HS042}, and using Corollary~\ref{corCap} it follows that $C(\chG,\Ninmod)=C_\chi(\chPS,\Ninmod)$.
\section{Gaussian capacities}For experimental implementations and analytical calculations, it is convenient to focus on Gaussian encodings. We call the capacity restricted to Gaussian encodings the Gaussian capacity $\CapG$ \cite{LPM09,SDKC09,PLM09,SKC10,SKC11,SKC12}:
\begin{equation}\label{eq:capgaussdef}
	\begin{split}
		\CapG(\chG,\Ninmod) & = \lim_{n \to \infty}\frac{1}{n}\CapGchi(\chG^{\otimes n},n\Ninmod),\\
		\CapGchi(\chG,\Ninmod) & = \max_{\muG \, : \, \rhoinmodG \in {\setN}^{G}_{\Ninmod}}\chi(\chG,\muG),
	\end{split}			
\end{equation}
where $\CapGchi(\chG,\Ninmod)$ is the one-shot Gaussian capacity. The maximum is now taken over all probability measures $\muG(\disp,\CM)$ on Gaussian symbol states $\rhoG(\disp,\CM)$ such that $\rhoinmodG(\dispinmod,\CMinmod) = \int \muG(d \disp,d \CM) \rhoG(\disp,\CM)$ is in the set $\setN^{G}_{\Ninmod}$ of Gaussian states with a mean photon number not greater than $\Ninmod$. Unlike previous works (e.g., \cite{H05,GaussRev2}), we require the individual symbol states as well as the average state to be Gaussian. Then we prove that the one-shot Gaussian capacity of an arbitrary single-mode Gaussian channel $\chG$ is given by the well-known expression \cite{EW07} (see \cite{SUPPMAT} for the proof)
\begin{equation}\label{eq:capgchi}		
	\begin{split}
			\CapGchi(\chG,\Ninmod) & = \max_{\CMin,\CMmod}\{\chiG(\seoutmod,\seout) \, | \, \Tr[\CMin+\CMmod] \le 2\Ninmod+1\},\\
			\chiG(\seoutmod,\seout) & = g\left(\seoutmod-\frac{1}{2} \right)-g\left(\seout-\frac{1}{2} \right),\\
			g(x) & = (x+1)\log_2(x+1)-x\log_2 x,
	\end{split}
\end{equation}
where $\CMin$ is the CM of a pure Gaussian input state $\rhoG(0,\CMin)$ satisfying $\det{(2\CMin)}=1$. Here $\CMmod$ is the CM of a classical Gaussian distribution according to which the input state is displaced in order to generate the modulated input state $\rhoinmodG(0,\CMinmod)$ with CM $\CMinmod=\CMin+\CMmod$ satisfying $\Tr[\CMinmod] \le 2\Ninmod+1$. Furthermore,  $\seout=\sqrt{\det \CMout}$ and $\seoutmod=\sqrt{\det \CMoutmod}$ are the symplectic eigenvalues of the output and modulated output states with CM $\CMout=\chG(\CMin)$ and $\CMoutmod=\chG(\CMinmod)$, respectively (see \cite{SUPPMAT}).

The one-shot Gaussian capacity is equal to the Gaussian capacity, i.e., $\CapG(\chG,\Ninmod)=\CapGchi(\chG,\Ninmod)$, provided it is additive. Interestingly, such an additivity can be proven if the input energy exceeds some threshold $\Nthr$ (see \cite{SUPPMAT}). Note that \cite{H05} also derives additivity but for a slightly different definition of $\CapGchi$ and without respecting the energy constraint. In addition, an analog of Corollary~\ref{corCap} can easily be shown to hold for Gaussian capacities, namely $\CapG(\chG,\Ninmod) = \CapG(\chPS,\Ninmod)$. Therefore, using the fiducial channel $\chPS$, we can analytically find the Gaussian capacity of any Gaussian channel in this high-energy regime:
\begin{corCapG}\label{corCapG} 
	For a single-mode Gaussian channel $\chG$ with parameters $(\tau \ne 0, \sey >0)$, there exists a fiducial channel $\chPS$ as defined in Theorem 1, such that
	\begin{eqnarray}
		& & \CapG(\chG,\Ninmod) = \CapG(\chPS_{(\tau,\sey,s)},\Ninmod)\nonumber\\
		& & =  g\left(|\tau|\Ninmod+\sey\cosh(2s)+\frac{|\tau|-1}{2} \right) - g\left(\sey+\frac{|\tau|-1}{2}\right),%
		\nonumber\\
		& & \mathrm{if~~} \Ninmod \ge \Nthr = \frac{1}{2} \left( e^{2|s|}+\frac{2\sey}{|\tau|} \sinh(2|s|) - 1 \right).\label{eq:CapGthr}
	\end{eqnarray}
\end{corCapG}
The proof is presented in \cite{SUPPMAT}. Note that the energy threshold $\Nthr$ depends on the parameter $s$ characterizing the fiducial channel $\chPS_{(\tau,\sey,s)}$. For thermal channels $\chPI = \chPS_{(\tau,y,0)}$, the threshold $\Nthr=0$, so that additivity holds in the entire energy range. Then, Eq.~\eqref{eq:CapGthr} coincides with previously derived expressions for particular cases \cite{Channels1,HW01}. In Fig.~\ref{fig:entbreak}, we illustrate an example of the domain where $\Ninmod \ge \Nthr$, i.e. Eq.~\eqref{eq:CapGthr} holds. Note, that Eq.~\eqref{eq:CapGthr} becomes the actual capacity $C(\chG,\Ninmod)$ (for $\Ninmod \ge \Nthr$) of an arbitrary single-mode Gaussian channel $\chG$ provided that the vacuum state is proven to minimize the output entropy of a single use of an ideal amplification channel \cite{GNLSC12,GC2013}.
\section{Upper bounds}Recently, upper bounds have been derived on the capacity of phase-insensitive channels, i.e. $\chPI$ with $\tau \ge 0$ \cite{KS13,GLMS12}. 
Using Corollary~\ref{corCapG}, we can generalize them to any Gaussian channel in the high-energy regime:
\begin{corCapBound}\label{corCapGBounds} 
	For a single-mode Gaussian channel $\chG$ with parameters $(\tau > 0, \sey>0)$ and $\Ninmod \ge \Nthr$, 
	\begin{equation}\label{eq:CapGBound}
		\begin{split}
			\CapG(\chG,\Ninmod) & \le C(\chG,\Ninmod) \le \overline{C} \le \CapG(\chG,\Ninmod) + \frac{1}{\ln 2},\\
			\overline{C} & = g\left(\frac{2\tau \Ninmod + (2\sey+1-\tau)\sinh^2{s}}{2\sey+1+\tau} \right),
		\end{split} 
	\end{equation}
	where $\CapG(\chG,\Ninmod)$ is stated in Eq.~\eqref{eq:CapGthr}.
	\begin{proof}
		The fiducial channel corresponding to $\chG$ can be decomposed as $\chPS_{(\tau,\sey,s)} = \chPS_{\left(G,\frac{G-1}{2},s\right)} \circ \chPS_{\left(T,\frac{1-T}{2},s\right)}$  with $T=2\tau/(2\sey+\tau+1)$ [see Fig.~\ref{fig:graph_ad} and Table~\ref{table:channels}]. Then, the capacity of $\chPS_{(\tau,\sey,s)}$ is upper bounded by the capacity of the first channel, i.e.
		\begin{equation*}
			C(\chG,\Ninmod) = C\left(\chPS_{\left(\tau,\sey,s\right)},\Ninmod\right) \le C\left(\chPS_{\left(T,\frac{1-T}{2},s\right)},\Ninmod\right) \le \overline{C}, 
		\end{equation*}
		where $\overline{C}=g(T\Ninmod + (1-T)\sinh^2{s})$ \cite{LPM09}. We define
		\begin{equation*}
			\Delta(s) \equiv \overline{C} - \CapG = g\bigl[A(B+1)^{-1}\bigr] - g(A+B\cosh^2{s}) + g(B),
		\end{equation*}
		where $A=\tau\Ninmod+\bigl[y-\frac{(\tau-1)}{2}\bigr]\sinh^2{s}$ and $B = y+\frac{\tau-1}{2}$. It was shown in \cite{KS13} that $\Delta(0) < 1/\ln 2$. Since $\forall s$ : $\Delta(s) \le \Delta(0)$, the corollary is proven.
	\end{proof}
\end{corCapBound}
Note that for $\tau < 0$ we can state a similar upper bound on the capacity, $C(\chG,\Ninmod) \le \overline{C}$, where $\overline{C}$ is given by Eq.~\eqref{eq:CapGBound} with the replacement $y \to -y$. However, in this case the last inequality in Eq.~\eqref{eq:CapGBound} does not hold. In a similar fashion, we extend in \cite{SUPPMAT} the bounds that were derived in \cite{GLMS12}.
\section{Conclusions}We have shown that an arbitrary single-mode Gaussian channel is either equivalent to a newly defined fiducial channel preceded by a phase shift and followed by a Gaussian unitary, or can be obtained in a proper limit of this combination. This equivalence was exploited to reduce the energy-constrained classical capacity of any single-mode Gaussian channel to that of the fiducial channel. We gave an analytical expression for the Gaussian capacity above the energy threshold, where additivity can be proven, and showed that in this case the classical capacity cannot exceed it by more than $1/\ln 2$ bits. We expect that our results will be useful for further studies on the capacities of Gaussian channels, especially for input energies below the energy threshold.

J.S. is grateful to Vittorio Giovannetti for clarifications on his work and acknowledges a financial support from the Belgian FRIA foundation. The authors also acknowledge financial support from the F.R.S.-FNRS under the Eranet project HIPERCOM, from the Interuniversity Attraction Poles program of the Belgian Science Policy Office under Grant No. IAP P7-35 ``Photonics$@$be'', from the Brussels Capital Region under the project CRYPTASC, from the ULB under the program ``Ouvertures internationales'', and from the Alexander von Humboldt Foundation.

\begin{widetext}
\section{Supplemental Material}
The quadrature operators $(\hat{q},\hat{p})$ of the quantized electromagnetic field mode are related to the annihilation and creation operator $\hat{a}$ and $\hat{a}^\dagger$, respectively, according to $\hat{q}=(\hat{a}+\hat{a}^\dagger)/\sqrt{2}$ and $\hat{p}=i(\hat{a}^\dagger-\hat{a})/\sqrt{2}$ and satisfy the canonical commutation relation $[\hat{q},\hat{p}]=i$. For a quantum system of $n$ modes we group them as $\bm{\hat{R}}=\trans{(\hat{q}_1,\hat{p}_1,\hat{q}_2,\hat{p}_2,...,\hat{q}_n,\hat{p}_n)}$. Then, for an $n$-mode quantum state $\hat{\rho}$ the coherent vector is given by $\disp = \braket{\bm{\hat{R}}}$ and the $2n \times 2n$ covariance matrix (CM) $\CM$ has entries $V_{ik}=\frac{1}{2}\braket{\{\hat{R}_i-\braket{\hat{R}_i},\hat{R}_k-\braket{\hat{R}_k}\}}$, where $\{,\}$ is the anti-commutator and $\braket{\hat{R}_i}=\Tr[\hat{R}_i\hat{\rho}]$. 

Let $\rhoG(\disp,\CM)$ be an $n$-mode Gaussian state with coherent vector $\disp$ and CM $\CM$. Then, a Gaussian channel $\chG$ is a completely-positive trace-preserving map which is closed on the set of Gaussian states~\cite{HW01}. It transforms input states with moments $\{\dispin,\CMin\}$ to output states with moments $\{\dispout,\CMout\}$ according to
\begin{equation}
	\dispout  =\X \dispin + \bm{\delta},\quad
	\CMout  =\X \CMin \trans{\X} + \eff,
\end{equation}
where $\bm{\delta}$ is the displacement introduced by the channel, $\X$ is a $2n \times 2n$ real matrix, and $\eff$ is a $2n \times 2n$ real, symmetric, and non-negative matrix. As mentioned in the main text we choose $\bm{\delta}=0$ in what follows (since the capacity is not affected by $\bm{\delta}$), and focus on the action of the map $\chG$ on second-order moments using the simplified notation $\chG(\CMin) = \CMout$. Then, the map $\chG$ is fully characterized by matrices $\X$ and $\eff$, which must satisfy $\eff+\frac i2\left(\bm{\Omega}-\X \bm{\Omega} \trans{\X}\right) \ge 0$ \cite{HW01}, where
\begin{equation}
	\bm{\Omega}=\bigoplus_{k=1}^n \begin{pmatrix} 0 & 1\\-1 & 0 \end{pmatrix},
\end{equation}
is the symplectic form~\cite{DeGosson}. For the one-mode case $n=1$ we define the parameters 
\begin{equation}\label{eq:tauy}
	\tau=\det{\X}, \quad y=\sqrt{\det{\eff}},
\end{equation}
which have to satisfy
\begin{equation}\label{eq:XYphysical}
	y \ge \frac{|\tau-1|}{2},
\end{equation}
in order for the map to be physical.
\subsection{Physical representation of non-thermal channels $\chCS$ and $\chSQ$}
Out of the seven canonical channels (see main text) the sixth one is the \emph{classical signal} (or quadrature erasing) channel, which we denote by $\chCS$. Its action is defined by
\begin{equation}\label{eq:A2}
	\XCS = \begin{pmatrix}
		1 & 0\\0 & 0 
	\end{pmatrix}, \quad
	\effCS = \begin{pmatrix}
		y & 0\\0 & y
	\end{pmatrix},
	\quad y \ge \frac{1}{2}.
\end{equation}
This channel can be physically implemented with a continuous-variable controlled-NOT (CV-NOT) gate \cite{FCP041,FCP042}. The corresponding scheme is depicted in Fig.~\ref{fig:graph_bc} (b), where $G~=~y+1/2$. Note that $\tau=0$, implying that $\chCS$ is always entanglement breaking (see Fig.~\ref{fig:entbreak} in the main text) and therefore, its classical capacity is additive \cite{S02HS041,S02HS042}.

The seventh and last canonical channel is the \emph{single-quadrature classical noise channel}, which we name $\chSQ$. Its action is defined by the matrices
\begin{equation}\label{eq:B1}
	\XSQ = \mathbb{I}, \quad
	\effSQ = \begin{pmatrix}
		0 & 0\\0 & \frac{1}{2}
	\end{pmatrix},
\end{equation}
where $\mathbb{I}$ is the $2 \times 2$ identity matrix. Equation~\eqref{eq:B1} implies $\tau=1$ and $y=0$ (as the perfect transmission channel). The channel $\chSQ$ is not entanglement breaking \cite{H08}.

\begin{figure}[h]
	\centering
	\includegraphics[width=0.49\textwidth]{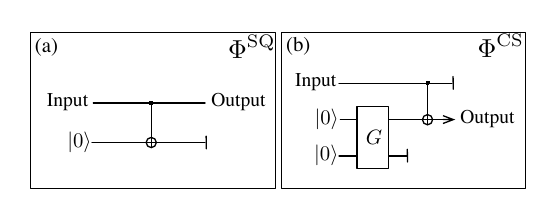}
	\caption{Realization of (a)  the classical signal channel $\chCS$ and (b) the single-quadrature classical noise channel $\chSQ$ by a beamsplitter with transmissivity $T$ and a two-mode squeezer with gain $G$. Here $\ket{0}$ stands for the vacuum state and ``$\dashv$'' denotes ``tracing out'' the mode.}
	\label{fig:graph_bc}
\end{figure}
 
In the following we explain the physical schemes of the channels $\chSQ$ and $\chCS$ as depicted in Fig.~\ref{fig:graph_bc} (a) and (b). First, we discuss the main ``building block''  of these schemes, namely the CV-CNOT gate acting on a two-mode state with CM $\CMinth \equiv \CMin \oplus \CM_{\rm th}$, consisting of a general input mode with CM $\CMin$ and an ancillary mode being in a thermal state with CM $\CM_{\rm th}$, i.e.
\begin{equation}
\CMin=
\begin{pmatrix}
	v_q 	& v_{qp}\\
	v_{qp}	& v_p
\end{pmatrix},\quad
\CM_{\rm th}=
\begin{pmatrix}
	y & 0\\
	0 & y
\end{pmatrix}.
\end{equation}
In the following the input mode corresponds to the ``control mode'' of the CV-CNOT gate, whereas the output mode is either the control or target mode depending on the channel. The action of the symplectic transformation $\Scnot$ (corresponding to the CV-CNOT gate) on the joint state reads \cite{FCP041,FCP042}
\begin{equation}\label{eq:SCNOT}
\Scnot \, \CMinth \, \trans{\Scnot} = 
\begin{pmatrix}
	1 & 0 & 0 & 0\\
	0 & 1 & 0 & -1\\ 
	1 & 0 & 1 & 0\\ 
	0 & 0 & 0 & 1
\end{pmatrix}
\begin{pmatrix}
	v_q 	& v_{qp} & 0 & 0\\
	v_{qp}	& v_p	 & 0 & 0\\
	0		& 0		 & y & 0\\
	0		& 0		 & 0 & y
\end{pmatrix}
\begin{pmatrix}
	1 & 0 & 1 & 0\\
	0 & 1 & 0 & 0\\ 
	0 & 0 & 1 & 0\\ 
	0 & -1 & 0 & 1
\end{pmatrix}
= 
\begin{pmatrix}
   v_q 	   		& v_{qp} 		& v_q 		& 0\\
   v_{qp}  		& v_p+y 		& v_{qp} 	& -y\\
   v_q			& v_{qp}		& y+v_q		& 0\\
   0			& -y			& 0			& y
\end{pmatrix}.
\end{equation}
Tracing out the target mode in Eq.~\eqref{eq:SCNOT} [as shown in Fig.~\ref{fig:graph_bc} (a)] and taking the ancillary mode to be in the vacuum state ($y=1/2$) leads to the output of the single quadrature additive noise channel $\chSQ$. In case of the channel $\chCS$ the ancillary mode with CM $\CM_{\rm th}$ is the output of the two-mode squeezer with gain $G$ [Fig.~\ref{fig:graph_bc} (b)], therefore, $y=G-1/2$. Tracing out the control mode in Eq.~\eqref{eq:SCNOT} [as shown in Fig. ~\ref{fig:graph_bc} (b)] leads to the output for the classical signal channel $\chCS$.

\subsection{Proof of Theorem~\ref{thmmapeqSup} and method how to obtain new decomposition}
\begin{figure}[h]
\centering
\includegraphics[width=0.4\textwidth]{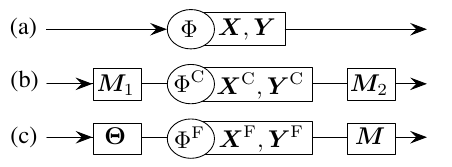}
\caption{Equivalence of (a) an arbitrary Gaussian channel $\chG$, (b) the canonical decomposition containing a canonical channel $\chH$ and (c) the decomposition in terms of the fiducial channel $\chPS$ as stated in Theorem~\ref{thmmapeqSup}.}
\label{fig:chEquiv}
\end{figure}
\begin{thmmapeqSup}\label{thmmapeqSup}	
		For a single-mode Gaussian channel $\chG$ defined by matrices $\X$ and $\eff$ with $\tau\ne 0$ and $y >0$, there exists a fiducial channel $\chPS$ defined by matrices $\XPS(\tau)=\sqrt{|\tau|}\diag(1,\sign(\tau))$, $\effPS(y,s)=y \, \diag{(e^{2s},e^{-2s})}$ with $\tau$ and $y$ obtained from Eq.~\eqref{def-tau-y} of the main text, a symplectic transformation $\SUchG$, and a rotation in phase space $\SrotchG$ such that
	\begin{equation}	
		\X = \SUchG \, \XPS(\tau) \, \SrotchG, \; \eff = \SUchG \, \effPS(y,s) \, \SUchGtrans,     \label{eq:mapeqPISup}
	\end{equation}	
		where the explicit dependencies of $\SUchG$, $\SrotchG$, and $s$ on the parameters of the channel $\chG$ are presented in Eqs.\eqref{eq:effbyeffPI}-\eqref{eq:SrotchG}.
\begin{proof}
The action of the single-mode Gaussian channel $\chG$ on an input CM $\CMin$ reads according to Eq.~\eqref{eq:chgauss}
\begin{equation}\label{eq:chGXY}
	\chG(\CMin) = \CMout =  \X \CMin \trans{\X} + \eff,
\end{equation}
where $\X$ is a real $2\times 2$ matrix and $\eff$ a real, symmetric and non-negative $2 \times 2$ matrix. In \cite{H07CGH061,H07CGH062,H07CGH063} it was stated that for any Gaussian channel $\chG$ there exists a canonical decomposition $\chG=\Utwo \circ \chH \circ \Uone$, where $\chH$ is a map belonging to one of the seven canonical types that are stated in the main text in Table~\ref{table:channels}. The corresponding action on the CM reads 
\begin{equation}\label{eq:chGcan}
	\CMout = \SUtwo(\XH \SUone \CMin \SUonetrans  \XH  + \effH)\SUtwotrans,
\end{equation}
where $\XH,\effH$ are the matrices defining the canonical channels (see Table~\ref{table:channels} in the main text) and $\SUone$, $\SUtwo$ are matrices corresponding to symplectic transformations realizing unitaries $\Uone$, $\Utwo$. In the following we obtain the new decomposition in terms of the fiducial channel as stated in the Theorem and furthermore, confirm Eq.~\eqref{eq:chGcan}. The proof is structured as follows. For given matrices $\X,\eff$ we have to distinguish three cases which depend on the ranks of $\X$ and $\eff$ and correspond to canonical decompositions for which $\chH$ is either $\chPI$, $\chSQ$ or $\chCS$. In the first case our new decomposition will contain finite squeezing operations, while for the other two cases the new decomposition is shown to be valid in a proper limit of infinite squeezing. 

For the following calculations we define the symplectic matrices corresponding to a rotation and a squeezing operation, i.e.
\begin{equation}\label{eq:SU}
	\Sgenrot(\theta)=\begin{pmatrix}
	\cos{\theta} & -\sin{\theta}\\ \sin{\theta} & \cos{\theta}
	\end{pmatrix},
\quad \Ssq(s)=\begin{pmatrix}e^{s} & 0\\0 & e^{-s}\end{pmatrix}.
\end{equation}
We sometimes omit the explicit dependence on the rotation angle or squeezing parameter. For the given CM $\eff$ there exists a rotation $\Srot_Y=\Sgenrot(\theta_Y)$, such that $\trans{\Srot}_Y \eff \Srot_Y = \diag{(y_1,y_2)}$, where $y_1, y_2 \ge 0$ are the eigenvalues of $\eff$. Since matrix $\X$ is real it has a singular value decomposition (SVD)
\begin{equation}\label{eq:XSVD}
	\X = \SrotXone \Lambda_X \bm{J} \SrotXtwo,
\end{equation} 
where $\SrotXone=\Sgenrot(\thetaXone)$, $\SrotXtwo=\Sgenrot(\thetaXtwo)$ and
\begin{equation}
	\Lambda_X = \diag{(x_1,x_2)}, \quad \bm{J} = \left\{ 
	  \begin{array}{l l}
	    \mathbb{I} & \quad \text{if $\tau \ge 0$}\\
	    \sigma_z & \quad \text{if $\tau < 0$}
	  \end{array} \right..
\end{equation}
Here, $x_1,x_2 \ge 0$ are the singular values and $\sigma_z = \diag{(1,-1)}$. Using equality $\det{\X}=\det{(\Lambda_X \bm{J})}$ and Eq.~\eqref{eq:tauy} we get $\tau=\pm x_1 x_2$ and $y=\sqrt{y_1 y_2}$. The condition on the determinants of $\X$ and $\eff$ stated in Eq.~\eqref{eq:XYphysical} allows us to exclude the following combinations of ranks because they are non-physical: $(\rank{\X},\rank{\eff}) \notin \{(0,0),(0,1),(1,0),(1,1)\}$. The physically allowed combinations of ranks therefore read $(\rank{\X},\rank{\eff}) \in \{(2,2),(0,2),(2,0),(1,2),(2,1)\}$. Our theorem corresponds to the case $(\rank{\X},\rank{\eff})=(2,2)$ which we prove at first. Then, we extend it to the second and third ``physical'' couple and finally treat the last two individually.

We begin with the case that is stated in the Theorem, i.e. $\rank{\X}=\rank{\eff}=2$. The latter implies that $x_1,x_2,y_1,y_2 \ne 0$. Then we can construct the squeezing operation $\Ssq_Y=\Ssq(s_Y)$, with $s_Y = \frac{1}{4}\ln{(y_1/y_2)}$ such that $\Ssq_Y^{-1}\diag{(y_1,y_2)}\Ssq_Y^{-1}=\diag{(y,y)}$. This implies that
\begin{equation}\label{eq:effbyeffPI}
	 \eff = \Srot_Y \Ssq_Y \effPI \Ssq_Y \trans{\Srot}_Y = y\Srot_Y \Ssq_Y^2 \trans{\Srot}_Y, \quad \effPI = \diag{(y,y)}.
\end{equation}
Here the symplectic transformation $(\Srot_Y \Ssq_Y)^{-1}$ realizes the symplectic diagonalization of $\eff$, where $y$ is the symplectic eigenvalue. Furthermore, we can define a squeezing operation $\Ssq_X=\Ssq(s_X)$, with $s_X = \frac{1}{2}\ln{(x_1/x_2)}$, such that Eq.~\eqref{eq:XSVD} can be written as
\begin{equation}\label{eq:XbyXPI}
	\X = \SrotXone \Ssq_X \XPI \SrotXtwo, \quad \XPI = \begin{pmatrix}
		\sqrt{|\tau|} & 0\\0 & \sign(\tau) \sqrt{|\tau|} 
	\end{pmatrix}.
\end{equation}
Notice that the matrix $\XPI$ has the property that 
\begin{equation}\label{eq:XPIrot}
	\XPI\Sgenrot(\theta) = \Sgenrot(\sign{(\tau)}\theta)\XPI.
\end{equation}
Now we obtain the decomposition $\eff = \SUchG  \effPS \SUchGtrans$ in the following way. We define
\begin{equation}\label{eq:SUchG}
	\SUchG = \SrotXone \Ssq_X \SrotPStrans,
\end{equation}
where $\SrotPS$ will be determined in the following. Then, we multiply $\eff$ in Eq.~\eqref{eq:effbyeffPI} from both sides with the identity matrix $\mathbb{I}=\SUchG \SrotPS \Ssq^{-1}_X \SrotXonetrans$,
\begin{equation}\label{eq:YPSderive}
	\mathbb{I}\eff\mathbb{I} = y \SUchG \SrotPS \Ssq^{-1}_X \SrotXonetrans \Srot_Y \Ssq^2_Y \trans{\Srot}_Y \SrotXone \Ssq^{-1}_X \SrotPStrans \SUchGtrans.
\end{equation}
Now we define
\begin{equation}\label{eq:effPS}
	\effPS = y\SrotPS \Ssq^{-1}_X \SrotXonetrans \Srot_Y \Ssq^2_Y \trans{\Srot}_Y \SrotXone \Ssq^{-1}_X \SrotPStrans,
\end{equation}
and thus, obtain the desired decomposition $\eff = \SUchG  \effPS \SUchGtrans$. Moreover, we chose the rotation $\SrotPS$ in a way such that matrix $\effPS$ is diagonal, i.e.
$\effPS	= y \, \diag{(e^{2s},e^{-2s})}$. This implies the following expression for the squeezing parameter $s$
\begin{equation}\label{eq:YrotPS}
	\begin{split}
		s & = \frac{1}{2}\ln\left[\frac{1}{4}e^{-2(s_X+s_Y)}(\xi-\sqrt{-16e^{4(s_X+s_Y)}+\xi^2})\right],\\
		\xi & = (1+e^{4s_Y})(1+e^{4s_X})-(-1+e^{4s_Y})(-1+e^{4s_X})\cos(2(\theta_Y-\thetaXone)).		
	\end{split}
\end{equation}
The angle $\theta_{\rm F}$ of rotation $\SrotPS=\Sgenrot(\theta_{\rm F})$ reads
\begin{equation}
	\theta_{\rm F} = -\mathrm{arcsin}\left(\frac{\sign{(\lambda)}}{\sqrt{1+\lambda^2}}\right),
\end{equation}
where
\begin{equation}
	\begin{split}
		\lambda & = -\frac{e^{-2s_X}(\tilde{\xi} + \sqrt{-16e^{4(s_X+s_Y)}+\xi^2})}{2\sin(2(\theta_Y-\thetaXone))(-1+e^{4s_Y})},\\
		\tilde{\xi} & = (1+e^{4s_Y})(-1+e^{4s_X})-(-1+e^{4s_Y})(1+e^{4s_X})\cos(2(\theta_Y-\thetaXone)).		
	\end{split}
\end{equation}
Using definition $\XPS=\sqrt{|\tau|}\diag(1,\sign(\tau))$ and Eq.~\eqref{eq:XPIrot} one can rewrite Eq.~\eqref{eq:XbyXPI} as
\begin{equation}\label{eq:XPSderive}
	\X = \SrotXone \Ssq_X \SrotPStrans \SrotPS \XPS \SrotXtwo = \SrotXone \Ssq_X \SrotPStrans \XPS \Sgenrot(\sign{(\tau)}\theta_{\rm F}) \SrotXtwo = \SUchG \XPS \SrotchG,
\end{equation}
where
\begin{equation}\label{eq:SrotchG}
	 \SrotchG = \Sgenrot(\sign{(\tau)}\theta_{\rm F}+\thetaXtwo).
\end{equation}
In summary, we found matrices $\SUchG$, $\SrotchG$ and the explicit parameters of $\XPS$ and $\effPS$ such that
\begin{equation}\label{eq:chPIbychPS}
	\X = \SUchG \XPS \SrotchG, \quad \eff = \SUchG  \effPS \SUchGtrans,
\end{equation}
and thus, we have proven the Theorem. 
\end{proof}
\end{thmmapeqSup} 

Now let us extend Theorem~\ref{thmmapeqSup} to other combinations of ranks.

\underline{$\rank{\X}=2, \rank{\eff}=0$:} Since $\eff=0$ it follows that $y=0$, which together with Eq.~\eqref{eq:XYphysical} implies that $\tau=1$. Note that the channel is unitarily equivalent to the perfect transmission channel. All relations derived above are found in the same way where one has to fix $s_Y=\theta_Y=0$, which leads to $\Ssq_Y=\Srot_Y=\mathbb{I}$.

\underline{$\rank{\X}=0, \rank{\eff}=2$:} This case can also be treated using the above relations. Since $\X=0$ it follows that $\tau=0$, which together with Eq.~\eqref{eq:XYphysical} implies that $y \ge 1/2$. This channel is unitarily equivalent to the zero-transmission channel and has trivially a capacity equal to zero. The decomposition containing the fiducial channel is found as above where one has to fix $s_X=\thetaXone=\thetaXtwo=0$.

We remark that for $(\rank{\X},\rank{\eff}) \in \{(2,2),(0,2),(2,0)\}$ the physical action of $\chG$ corresponds (up to unitaries) to the action of $\chPI$. Indeed, by inserting Eqs.~\eqref{eq:effbyeffPI} and \eqref{eq:XbyXPI} in Eq.~\eqref{eq:chGXY} one obtains the canonical decomposition $\chG = \Utwo \circ \chPI \circ \Uone$, which in terms of the symplectic transformations reads as in Eq.~\eqref{eq:chGcan}, with
\begin{equation}\label{eq:chPIcan}
	\XH=\XPI, \quad \effH=\effPI, \quad \SUone = \Ssq_Y^{-1}\trans{\Srot'_Y}{\SrotXone'} \Ssq_X \SrotXtwo, \quad \SUtwo = \Srot_Y \Ssq_Y,
\end{equation}	
where $\Srot_Y'=\Sgenrot(\sign(\tau)\theta_Y)$ and ${\SrotXone'}=\Sgenrot(\sign(\tau)\thetaXone)$. In Fig.~\ref{fig:chEquiv} we sketched the equivalences found above.

\underline{$\rank{\X}=2,\rank{\eff}=1$:} This implies $y=0$ and together with Eq.~\eqref{eq:XYphysical} that $\tau=1$. The eigenvalues of $\eff$ now read $y_1=0,y_2>0$ (the other case $y_1>0, y_2=0$ can be treated equivalently). Similarly to the case $\rank{\eff}=2$ one can find a rotation $\Srot_Y$ such that $\trans{\Srot}_Y \eff \Srot_Y = \diag{(0,y_2)}$. Then, one can construct a squeezing operation $\Ssq_Y$ with $s_Y=-\frac{1}{2}\ln(2y_2)$ which yields
\begin{equation}\label{eq:effbyeffSQ}
	\eff = \Srot_Y \Ssq_Y \effSQ \Ssq_Y \trans{\Srot}_Y, \quad \effSQ = \diag{\left(0,\frac{1}{2}\right)}.
\end{equation}
The matrix $\effSQ$ can be recovered with an additional squeezer $\Ssq_T=\Ssq(s_T)$ in the limit of infinite squeezing, i.e. $\effSQ = \lim_{s_T \to \infty}\frac{1}{2}e^{-2s_T}\Ssq^{-2}_T$ from which follows
\begin{equation}\label{eq:effSQ}
	\eff = \lim_{s_T \to \infty}\frac{1}{2}e^{-2s_T} \Srot_Y \Ssq_Y \Ssq^{-2}_T \Ssq_Y \trans{\Srot}_Y = \lim_{s_T \to \infty}\frac{1}{2}e^{-2s_T} \Srot_Y \Ssq_{YT}^2 \trans{\Srot}_Y,
\end{equation}
where $\Ssq_{YT}=\Ssq(s_Y-s_T)$. Since $\rank{\X}=2$ we can decompose $\X$ as in Eq.~\eqref{eq:XbyXPI} but with the simplification $\tau=1$, i.e.
\begin{equation}\label{eq:XbyXSQ}
	\X = \SrotXone \Ssq_X \XSQ \SrotXtwo, \quad \XSQ=\mathbb{I}.
\end{equation}
We observe that we can replace $\XSQ=\XPS$, where $\XPS$ is as defined as above with $\tau = 1$. Thus, we get the same decomposition as stated in Eq.~\eqref{eq:XbyXPI}. Now one can recover both matrices $\X,\eff$ as a limiting case of Eq.~\eqref{eq:chPIbychPS}, namely,
\begin{equation}
	\X = \lim_{s_T \to \infty}\SUchG \XPS \SrotchG, \quad \eff = \lim_{s_T \to \infty}\SUchG \effPS \SUchGtrans,
\end{equation}
where in the definitions of $\SUchG$ \eqref{eq:SUchG}, $\SrotchG$ \eqref{eq:SrotchG} and $\effPS$ \eqref{eq:effPS} one has to make replacements $\tau \to 1$, $s_Y \to s_Y-s_T$ and $y \to \frac{1}{2}e^{-2s_T}$. This replacement only affects matrix $\effPS$ and rotations $\SrotPS$ and $\SrotchG$. Thus, we recovered both matrices $\X$ and $\eff$ as a limiting case of the decomposition stated in the Theorem.

Note that the physical action of $\chG$ in this case corresponds (up to unitaries) to the action of $\chSQ$: by inserting Eqs.~\eqref{eq:XbyXSQ} and \eqref{eq:effbyeffSQ} into Eq.~\eqref{eq:chGXY} we recover the canonical decomposition $\chG = \Utwo \circ \chSQ \circ \Uone$, which in terms of the symplectic transformations is given by Eq.~\eqref{eq:chGcan}, with
\begin{equation}\label{eq:chSQcan}
   	\XH=\XSQ, \quad \effH=\effSQ, \quad \SUone = \Ssq_Y^{-1}\trans{\Srot}_Y\SrotXone \Ssq_X \SrotXtwo, \quad \SUtwo = \Srot_Y \Ssq_Y.
\end{equation}

\underline{$\rank{\X}=1,\rank{\eff}=2$:} Since in this case $\tau=0$, it follows from Eq.~\eqref{eq:XYphysical} that $y \ge \frac{1}{2}$. The SVD of $\X$ now reads $\X = \SrotXone \diag(x_1,0)$ (the other case $x_1=0,x_2>0$ can be treated equivalently). One can define $\Ssq_X=\Ssq(s_X)$ with $s_X=\ln(x_1)$ such that 
\begin{equation}\label{eq:XbyXCS}
	\X = \SrotXone \Ssq_X \XCS, \quad \XCS = \diag{(1,0)}.
\end{equation}
Since $\XCS$ can be expressed as $\XCS = \lim_{s_T \to \infty}e^{-s_T}\Ssq_T$, where $\Ssq_T=\Ssq(s_T)$, Eq.~\eqref{eq:XbyXCS} becomes
\begin{equation}\label{eq:XCS}
	\X = \lim_{s_T \to \infty}e^{-s_T} \SrotXone \Ssq_X \Ssq_T = \lim_{s_T \to \infty}e^{-s_T} \SrotXone \Ssq_{XT},
\end{equation}
where $\Ssq_{XT}=\Ssq(s_X+s_T)$. Since $y \ge \frac{1}{2}$ we find as in the case $y>0$ treated above [see derivation of Eq.~\eqref{eq:effbyeffPI}], a rotation $\Srot_Y$ and squeezing $\Ssq_Y$ such that 
\begin{equation}\label{eq:effbyeffCS}
	\eff = \Srot_Y \Ssq_Y \effCS \Ssq_Y \trans{\Srot}_Y, \quad \effCS=\diag{(y,y)}, \quad y \ge \frac{1}{2}.
\end{equation}
Thus, we recover matrices $\X, \eff$ as a limiting case of Eq.~\eqref{eq:chPIbychPS}, i.e.
\begin{equation}\label{eq:chCSbychPS}
	\X = \lim_{s_T \to \infty}\SUchG \XPS \SrotchG, \quad \eff = \lim_{s_T \to \infty}\SUchG \effPS \SUchGtrans,
\end{equation}
where in the definitions of $\SUchG$ \eqref{eq:SUchG} and $\SrotchG$ \eqref{eq:SrotchG} one has to make replacements $\thetaXtwo \to 0$, $s_X \to s_X+s_T$ and $\tau \to e^{-2s_T}$. Note that this replacement affects $\SUchG$ but does not affect matrix $\eff$ stated in Eq.~\eqref{eq:effbyeffCS}. Therefore, we found also for the last case both matrices $\X,\eff$ as limiting cases of the decomposition stated in the Theorem.

Now we demonstrate that (up to unitaries) the physical action of $\chG$ in this case corresponds to the action of $\chCS$. By inserting Eqs.~\eqref{eq:effbyeffCS} and \eqref{eq:XbyXCS} in Eq.~\eqref{eq:chGXY}, we obtain 
\begin{equation}
	\chG(\CMin) = \SUtilde(\tilde{\X} \CMin \trans{\tilde{\X}}  + \effCS)\SUtildetrans, \quad \tilde{\X} = \Ssq_Y^{-1}\trans{\Srot}_Y\SrotXone \Ssq_X \XCS, \quad \SUtilde = \Srot_Y \Ssq_Y.
\end{equation}
For the real $2\times 2$ matrix $\tilde{\X}$ one can again obtain the SVD which leads to $\tilde{\X}=\tilde{\Srot}_X \tilde{\Ssq}_X \XCS$. Since $\tilde{\Ssq}_X \XCS = \XCS \tilde{\Ssq}_X$ we obtain the canonical decomposition $\chG = \Utwo \circ \chCS \circ \Uone$ in terms of the symplectic transformations as stated in Eq.~\eqref{eq:chGcan}, with
\begin{equation}\label{eq:chCScan}
	\XH=\XCS, \quad \effH=\effCS,\quad \SUone = \tilde{\Ssq}_X, \quad \SUtwo = \Srot_Y \Ssq_Y \tilde{\Srot}_X,
\end{equation}
Thus, we extended the Theorem to lower rank cases of $\X$ and $\eff$. 

We remark that both channels $\chSQ$ and $\chCS$ are obtained by gradually increasing $s_T$ and since for each finite $s_T$ Corollary~\ref{corCap} (stated in the main text) holds, it also remains valid in the limit $s_T \to \infty$.
\subsection{Derivation of simplified expression for the one-shot Gaussian capacity}
In the following we show that the one-shot Gaussian capacity of a single-mode Gaussian channel $\chG$ can be expressed as
\begin{eqnarray}
	\CapGchi(\chG,\Ninmod) & = &\max_{\CMin,\CMmod}\{\chiG(\seoutmod,\seout) \; | \; \Tr[\CMin+\CMmod] \le 2\Ninmod+1\},\label{eq:CapGchiSup}\\
	\chiG & = & g\left(\seoutmod-\frac{1}{2} \right)-g\left(\seout-\frac{1}{2} \right),\\
	g(x) & = & (x+1)\log_2(x+1)-x\log_2 x,\label{eq:g}
\end{eqnarray}
where $g(0)=0$, $\CMin$ is the CM of a pure Gaussian input state fulfilling $\det{(2\CMin)}=1$, $\CMmod$ is the CM of a classical Gaussian distribution used to displace the input state and to generate the modulated input state with CM $\CMinmod=\CMin+\CMmod$ where $\Tr[\CMinmod] \le 2\Ninmod+1$. Furthermore,  $\seout=\sqrt{\det \CMout}$ and $\seoutmod=\sqrt{\det \CMoutmod}$ are the symplectic eigenvalues of the output and modulated output state with CM $\CMout=\chG(\CMin)$ and $\CMoutmod=\chG(\CMinmod)$, respectively.

Equation~\eqref{eq:CapGchiSup} states that among all possible Gaussian sources characterized by a measure $\muG(d \disp,d \CM)$ over the set of Gaussian states $ \rhoG(\disp,\CM)$
of mean $\disp$ an CM $\CM$, the source optimizing the Gaussian capacity corresponds to a single pure Gaussian state $ \rhoG(0,\CMin)$ with covariance $\CMin$ fulfilling $\det{(2\CMin)}=1$, modulated by phase-space translatations following a Gaussian distribution with CM $\CMmod$. 

To achieve our goal we use the fact that the maximization inside the Gaussian capacity definition
\begin{equation}
    \CapGchi(\chG,\Ninmod) = \max_{\muG \, : \, \rhoinmodG \in {\setN}^{G}_{\Ninmod}} \left[ S(\chG[\rhoinmodG]) - \int{\muG(d \disp,d \CM)\, S(\chG[\rhoG(\disp,\CM)])} \right],
\end{equation}
can be divided into two different steps. In the first step, among all the sources $\muG(d \disp,d \CM)$ belonging to the set ${\setMu}^{G}_{\rhoinmodG}$
sharing the same average input state
\begin{equation}
	\rhoinmodG \equiv \rhoG(0,\CMinmod) = \int \muG(d \disp,d \CM) \rhoG{(\disp,\CM)},
\end{equation}
we maximize the modified Holevo quantity
\begin{equation}
	\tilde{\chi}(\chG,\Ninmod,\rhoinmodG) = S(\chG[\rhoinmodG]) - \min_{\muG \in {\setMu}^{G}_{\rhoinmodG}} \int{\muG(d \disp,d \CM)\, S(\chG[\rhoG(\disp,\CM)])}.
\label{eq:OptHolevo}
\end{equation}
Note that the choice of zero mean for the average input state $\rhoinmodG$ is natural because displacements do not change the entropy, however require energy. In the second and final step we optimize $\tilde{\chi}(\chG,\Ninmod,\rhoinmodG)$ over the average input state $\rhoinmodG$ satisfying the energy constraint $\Ninmod$, thus obtaining $C^{G}_\chi(\chG,\Ninmod)$. We use the fact that the minimum of the average output entropy appearing in equation (\ref{eq:OptHolevo}) can be rewritten as the \emph{Gaussian entanglement of formation} $\EG[\bar{\sigma}_{\rm BE}]$ (see \cite{WGKWC04}), i.e.
\begin{equation}
	\min_{\muG \in {\setMu}^{G}_{\rhoinmodG}} \int{\muG(d \disp,d \CM)\, S(\chG[\rhoG(\disp,\CM)])}=\EG[\bar{\sigma}_{\rm BE}]
\label{eq:MinOutEntropy}
\end{equation}
of a given bipartite mixed state $\bar{\sigma}_{\rm BE}=U_\chG \rhoinmodG_B \otimes \ket{0}\bra{0}_E U_\chG^\dagger$ with CM $\CMBEmod$ obtained by the unitary (Stinespring) dilation $U_\chG$ of channel $\chG$, such that $\chG[\rhoinmodG]={\rm Tr}_E[\bar{\sigma}_{\rm BE}]$. Indeed, the Gaussian entanglement of formation is defined as
\begin{equation}
	\EG[\bar{\sigma}_{\rm BE}]=\EG[U_\chG \rhoinmodG \otimes \ket{0}\bra{0} U_\chG^\dagger] = \min_{\muG \in {\setMu}^{G}_{\rhoinmodG}} \left[ \int{\muG(d \disp, d \CM) E[U_\chG \rhoG(\disp,\CM) \otimes \ket{0}\bra{0} U_\chG^\dagger]}\right].
\end{equation}
Here, $E[U_\chG \rhoG(\disp,\CM) \otimes \ket{0}\bra{0} U_\chG^\dagger]$ is the entanglement of a bipartite Gaussian state $\sigma_{\rm BE}=U_\chG \rhoG(\disp,\CM) \otimes \ket{0}\bra{0} U_\chG^\dagger$ with CM $\CMBE$. The entanglement $E$ of a bipartite state is quantified by the von Neumann entropy of any of its two reduced density operators. Equation~\eqref{eq:MinOutEntropy} not only simplifies the capacity definition to
\begin{equation}
 		\CapGchi(\chG,\Ninmod) = \max_{\rhoinmodG \in {\setN}^{G}_{\Ninmod}} \left[ S(\chG[\rhoinmodG]) - \EG[\bar{\sigma}_{\rm BE}] \right],
\end{equation}
but also leads to the proof of Eq.~\eqref{eq:CapGchiSup}. According to \cite{WGKWC04} the Gaussian entanglement of formation can be simplified to
	\begin{equation}\label{eq:GEF}
		\EG[\bar{\sigma}_{\rm BE}] = \min_{\CMBE} \{ E[\sigma_{\rm BE}] \, | \, \CMBE \le \CMBEmod \},
	\end{equation}
where the minimum is taken over a single pure bipartite Gaussian state with CM $\CMBE$. This implies the existence of a covariance matrix $\CMin$ such that the output entropy of pure Gaussian states $\rhoG(\disp,\CMin)$ achieves the minimum in Eq.~\eqref{eq:MinOutEntropy} for any $\disp$. Due to Eq.~\eqref{eq:GEF} and the fact that the symplectic transformation which corresponds to $U_\chG$ does not change the positivity $\CMBEmod-\CMBE \ge 0$ it follows that $\CMinmod-\CMin \ge 0$. Then modulating $\rhoG(\disp,\CMin)$ according to a Gaussian distribution with covariance matrix $\CMmod=\CMinmod-\CMin$ generates a source with average input state $\rhoinmodG$ with CM $\CMinmod$ saturating the bound of Eq.~\eqref{eq:OptHolevo}. Thus, the one-shot Gaussian capacity of a general $n$-mode Gaussian channel $\chG$ simplifies to
\begin{equation}\label{eq:CapGchimultigen}
	\CapGchi(\chG,\Ninmod) = \max_{\CMin,\CMmod} \{ S(\chG(\CMin+\CMmod)) - S(\chG(\CMin)) \; | \; \Tr[\CMin+\CMmod] \le 2n\Ninmod+n\}.
\end{equation}
The entropy of an $n$-mode Gaussian state $\rhoG(\disp,\CM)$ can be calculated in terms of the $n$ symplectic eigenvalues $\nu_i$ of $\CM$, i.e.
\begin{equation}
	S(\rhoG(\disp,\CM))=\sum_{i=1}^n \, g\left(\nu_i-\frac{1}{2}\right),
\end{equation}
where $g(x)$ is defined in Eq.~\eqref{eq:g}. Therefore, the final expression for the one-shot Gaussian capacity reads
\begin{equation}\label{eq:CapGchimulti}
	\CapGchi(\chG,\Ninmod) = \max_{\CMin,\CMmod}\left\{ \left[\sum\limits_{i=1}^n g\left(\seoutmod_i - \frac{1}{2} \right) - g\left(\seout_i - \frac{1}{2} \right)\right] \; \Biggl| \; \Tr[\CMin+\CMmod] \le 2n\Ninmod+n \right\},
\end{equation}
where $\seout_i$ and $\seoutmod_i$ are the symplectic eigenvalues of the CM $\CMout=\chG(\CMin)$ and CM $\CMoutmod=\chG(\CMinmod)$. For the one mode case Eq.~\eqref{eq:CapGchimulti} simplifies to the well-known expression (see e.g. \cite{EW07}) stated in Eq.~\eqref{eq:CapGchiSup}.
\subsection{Proof of Corollary~\ref{corCapGSup}}\label{sec:corCapGSup}
\begin{figure}
\centering
\includegraphics[width=0.65\textwidth]{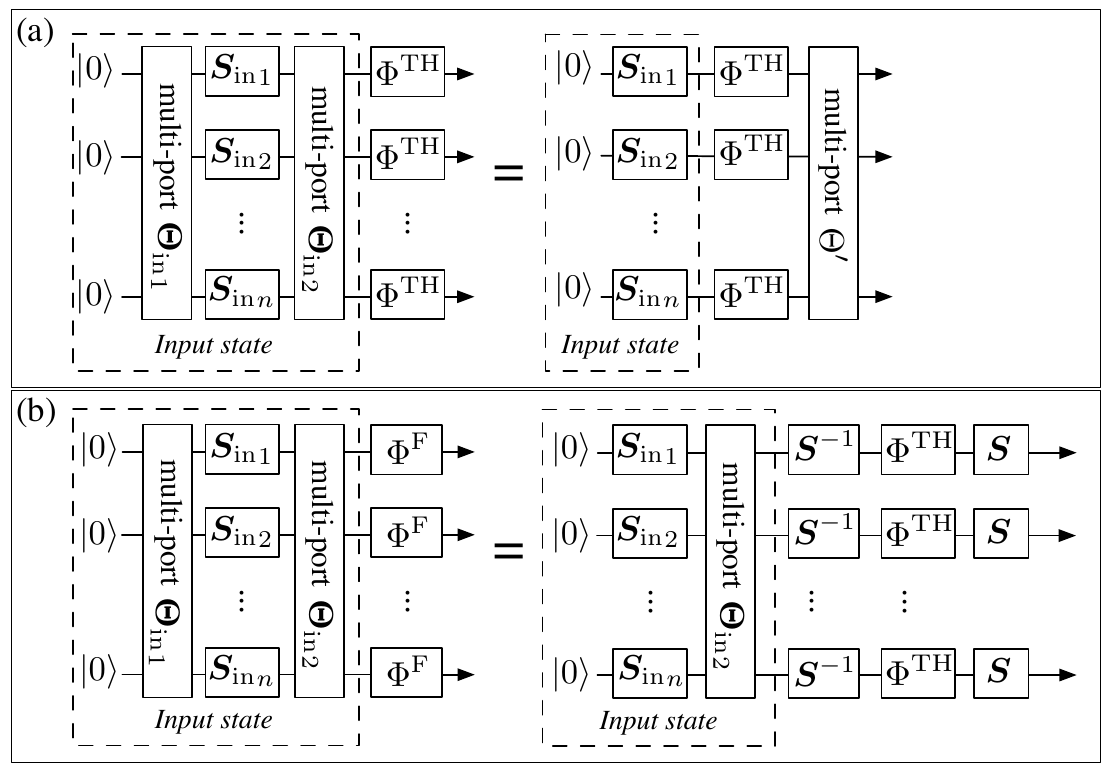}
\caption{Bloch Messiah decomposition: any multimode pure Gaussian state can be generated from the $n$-mode vacuum state and a set of single mode squeezers ${\Ssq_{\rm in}}_i$ preceded and followed by linear multi-port interferometers ${\Srot_{\rm in}}_1, {\Srot_{\rm in}}_2$. (a) Reduction of the scheme for a collection of channels $({\chPI})^{\otimes n}$ and (b) Reduction for the fiducial channel $(\chPS)^{\otimes n}$.}
\label{fig:chBlochMessiah}
\end{figure}
\begin{corCapGSup}\label{corCapGSup}
For a single-mode Gaussian channel $\chG$ with parameters $(\tau \ne 0, y >0)$, there exists a fiducial channel $\chPS$ as defined in Theorem~\ref{thmmapeqSup} such that
\begin{equation}\label{eq:CapGthrSup}
	\begin{split}
			& \CapG(\chG,\Ninmod) = \CapG(\chPS_{(\tau,y,s)},\Ninmod)= g\left(|\tau|\Ninmod+y\cosh(2s)+\frac{|\tau|-1}{2} \right) - g\left(y+\frac{|\tau|-1}{2}\right),\\
			& \mathrm{if~~} \Ninmod \ge \Nthr = \frac{1}{2} \left( e^{2|s|}+\frac{2y}{|\tau|} \sinh(2|s|) - 1 \right).
	\end{split}
\end{equation}
\end{corCapGSup}
\begin{proof}
The statement of Corollary~\ref{corCap} (see main text) formulated for the classical capacity can be straightforwardly extended to the Gaussian capacity, i.e. $\CapG(\chG,\Ninmod)=\CapG(\chPS,\Ninmod)$. This means that we only have to evaluate the Gaussian capacity of the fiducial channel $\chPS$ in order to find the Gaussian capacity of $\chG$. In the following we find $\CapG(\chPS,\Ninmod)$ explicitly for input energies $\Ninmod \ge \Nthr$. 

The proof is structured as follows. First, we prove that the Gaussian minimum output entropy of thermal channels $\chPI$ is additive (corresponding to $\chPS$ with $s=0$). Then we extend this proof to the fiducial channel for input energies $\Ninmod \ge \Nthr$ (we present a simple and physically motivated proof, which is an alternative to the one in \cite{H05}). Then we show that this also implies the additivity of the one-shot Gaussian capacity in this energy domain. Finally, we derive the exact expression for the one-shot Gaussian capacity.

In \cite{B05} it was shown that any pure $n$-mode Gaussian (input) state can be generated from the $n$-mode vacuum state, using $n$ single-mode squeezers ${\Ssq_{\rm in}}_i$ preceded and followed by a linear multi-port interferometer, corresponding to passive symplectic transformations ${\Srot_{\rm in}}_1$ and ${\Srot_{\rm in}}_2$ [see Fig.~\ref{fig:chBlochMessiah} (a)]. This decomposition can be further simplified since the $n$-mode vacuum state with CM $\mathbb{I}/2$ (where $\mathbb{I}$ is the $2n \times 2n$ identity matrix) remains unchanged under the action of the first interferometer ${\Srot_{\rm in}}_1$ and therefore, we can omit ${\Srot_{\rm in}}_1$ without changing the input state. The action of the channel ${(\chPI_{(\tau,y)})}^{\otimes n}$ in terms of symplectic transformations then reads
\begin{equation}\label{eq:CMoutBM}
	\CMout = \frac{1}{2}\XPI {\Srot_{\rm in}}_2 \Ssq_{\rm in} \mathbb{I} \Ssq_{\rm in} {\trans{\Srot}_{\rm in}}_2\XPI+ \effPI,
\end{equation} 
where $\Ssq_{\rm in}=\oplus_{i=1}^n{\Ssq_{\rm in}}_i$ corresponds to the joint operation of single mode squeezers, ${\Srot_{\rm in}}_2$ is the symplectic transformation of the second linear multi-port interferometer and 
\begin{equation}
	\XPI = \bigoplus_{i=1}^n \begin{pmatrix}
		\sqrt{|\tau|} & 0\\0 & \sign(\tau) \sqrt{|\tau|} 
	\end{pmatrix}, \quad \effPI = y \mathbb{I}.
\end{equation}
By inserting $\frac{1}{|\tau|}\XPI \XPI = \mathbb{I}$ between ${\Srot_{\rm in}}_2$ and $\Ssq_{\rm in}$ and between $\Ssq_{\rm in}$ and ${\trans{\Srot}_{\rm in}}_2$ in Eq.~\eqref{eq:CMoutBM} we obtain
\begin{equation}
	\begin{split}
		\CMout =  \Srot' \left(\frac{1}{2}\XPI \Ssq_{\rm in} \mathbb{I} \Ssq_{\rm in} \XPI + \effPI \right) \trans{\Srot'},
	\end{split}
\end{equation}
where $\Srot'=\frac{1}{|\tau|} \XPI {\Srot_{\rm in}}_2 \XPI$ and $\Srot' \effPI \trans{\Srot'} = \effPI$. One can confirm easily that $\Srot'$ is indeed a symplectic rotation matrix (corresponding to another interferometer at the output), i.e. that $\Srot' \bm{\Omega} \trans{\Srot'} = \bm{\Omega}$ and $\Srot' \trans{\Srot'}=\mathbb{I}$. Thus, the general Gaussian input state entering the channel $(\chPI)^{\otimes n}$ is reduced to a product state [right hand side of Fig.~\ref{fig:chBlochMessiah} (a)].

Now we calculate the entropy of the output state with CM $\CMout$. Since $\Srot'$ does not affect its entropy we can omit it. Therefore, we are left with $n$ vacuum modes entering the set of single-mode squeezers and then individually passing each channel $\chPI$. Therefore, the additivity of the (Gaussian) output entropy is proven and thus,
\begin{equation}
\min_{\CMin}\frac{1}{n}S(\CMout) =\min_{\CMin}\frac{1}{n}S\left({\left(\chPI\right)}^{\otimes n}\left(\frac{1}{2}\Ssq_{\rm in} \mathbb{I} \Ssq_{\rm in}\right)\right) =  
	\frac{1}{n}\sum\limits_i \min_{{\CMin}_i}S\left(\chPI({\CMin}_i)\right) = \frac{1}{n}\sum\limits_i \min_{{\CMin}_i} g\left({\seout}_i-\frac{1}{2}\right),
\end{equation}
where ${\CMin}_i=\frac{1}{2}{\Ssq^2_{\rm in}}_i$ is the CM of the mode $i$ that exits squeezer ${\Ssq_{\rm in}}_i$ and ${\seout}_i=\sqrt{\det{(\chPI({\CMin}_i))}}$ is the symplectic eigenvalue of the corresponding output state. The output entropy is minimized for ${\CMin}_i=\mathbb{I}/2, \forall i$, i.e. by removing all one-mode squeezers at the input. This leads to
\begin{equation}
	\min_{\CMin}\frac{1}{n}S(\CMout) = g\left(y+\frac{|\tau|-1}{2}\right).
\end{equation}
Now let us consider the same problem for the fiducial channel $\chPS$ as depicted in Fig.~\ref{fig:chBlochMessiah} (b), where we use again the Bloch-Messiah decomposition. The first interferometer ${\Srot_{\rm in}}_1$ can again be omitted because it does not affect the $n$-mode vacuum state. From the definition of the fiducial channel we have (for the case of one mode) the equivalence $\XPS=\XPI, \effPS = \Ssq \effPI \Ssq$, with $\Ssq=\diag(e^{s},e^{-s})$, where $\Ssq$ is the squeezing operation of the environment. This leads to the equality
\begin{equation}
	\chPS(\CMin)=\Ssq(\XPI \Ssq^{-1}\CMin \Ssq^{-1}\XPI + \effPI)\Ssq.
\end{equation}
As a consequence we can replace each fiducial channel by a thermal channel preceded by an anti-squeezer and followed by a squeezer [see right hand side of Fig.~\ref{fig:chBlochMessiah} (b)]. 

Now we focus again on the minimization of the output entropy. Then, the squeezers $\Ssq$ at the output of each channel $\chPI$ can be omitted since they do not change the entropy and we have no energy constraint on the output. We showed above that the entropy of the joint map ${(\chPI)}^{\otimes n}$ is minimized by the $n$-mode vacuum state. Thus, the multi-mode input state that minimizes the output entropy of the fiducial channel has to be in the $n$-mode vacuum state after passing the $n$ anti-squeezers $\Ssq^{-1}$ [see right hand side of Fig.~\ref{fig:chBlochMessiah} (b)]. 
Therefore, it is optimal to fix the input interferometer ${\Srot_{\rm in}}_2=\mathbb{I}$ and to chose each squeezer ${\Ssq_{\rm in}}_i$ at the input to ``undo'' each anti-squeezer, i.e. ${\Ssq_{\rm in}}_i=\Ssq, \, \forall i$. Thus, the $n$-mode Gaussian input state that minimizes the output entropy of the channel ${(\chPS)}^{\otimes n}$ is a product state with CM ${\CMin}=\oplus_{i=1}^n\Ssq^2/2$. It follows that the Gaussian minimum output entropy is additive and reads
\begin{equation}\label{eq:minSat}
	\min_{\CMin}\frac{1}{n}S\left({(\chPS)}^{\otimes n}(\CMin)\right) = g\left(y+\frac{|\tau|-1}{2}\right),
\end{equation}
where the minimization requires a certain amount of energy to undo each squeezer, which will be taken into account later.

We show now that Eq.~\eqref{eq:minSat} leads to the additivity of the one-shot Gaussian capacity of $\chPS$ for input energies $\Ninmod \ge \Nthr$. Using the expression of the one-shot Gaussian capacity $\CapGchi(\chPS,\Ninmod)$ stated in Eq.~\eqref{eq:CapGchimultigen} and using the definition of the Gaussian capacity, i.e.
\begin{equation}
	\CapG(\chG,\Ninmod) = \lim_{n \to \infty}\frac{1}{n}\CapGchi(\chG^{\otimes n},n\Ninmod),
\end{equation}
we can state the following upper bound:
\begin{equation}\label{eq:CapGup}
	\begin{split}
		\CapG(\chPS,\Ninmod) & \le \max_{\CMin,\CMmod}S\left(\chPS(\CMin+\CMmod)\right) - \lim_{n \to \infty}\min_{\CMin}\frac{1}{n}S\left({(\chPS)}^{\otimes n}(\CMin)\right),
	\end{split}	
\end{equation}
where the first term only needs to be maximized for a single use of the channel due to the subadditivity of the entropy. It is known that a thermal state maximizes the von Neumann entropy, therefore, the optimal modulated output state is a thermal state carrying the total number of photons, i.e. $\chPS(\CMin+\CMmod)=\CMoutmod=\diag{(\seoutmod,\seoutmod)}$, where
\begin{equation}\label{eq:seopt}
	\seoutmod= |\tau| \Ninmod + y\cosh(2s)+|\tau|/2.
\end{equation}
The second term in Eq.~\eqref{eq:CapGup} was already evaluated above [see Eq.~\eqref{eq:minSat}]. In summary, we found
\begin{eqnarray}
	\max_{\CMin,\CMmod}S\left(\chPS(\CMin+\CMmod)\right) & = & g\left(|\tau| \Ninmod + y\cosh(2s)+\frac{|\tau|-1}{2}\right),\label{eq:CapMax}\\
	\lim_{n \to \infty}\min_{\CMin}\frac{1}{n}S\left({(\chPS)}^{\otimes n}(\CMin)\right) & = & g\left(y+\frac{|\tau|-1}{2}\right).\label{eq:CapMin}		
\end{eqnarray}
The encoding which realizes both, the maximum and the minimum in Eqs.~\eqref{eq:CapMax} a \eqref{eq:CapMin} is given by
\begin{equation}\label{eq:maxminenc}
	\CMin=\Ssq^2/2, \quad,
	\CMmod=\CMinmod-\CMin, \quad
	\CMinmod=\begin{pmatrix}
		\Ninmod+\frac{1}{2}-\frac{y}{|\tau|}\sinh(2s) & 0\\0 & \Ninmod+\frac{1}{2}+\frac{y}{|\tau|}\sinh(2s)
	\end{pmatrix}.
\end{equation}
This encoding can only be realized if 
\begin{equation}
	\Ninmod \ge \Nthr= \frac{1}{2} \left( e^{2|s|}+\frac{2y}{|\tau|} \sinh(2|s|) - 1 \right),
\end{equation}
because $\Ninmod < \Nthr$ would imply $\CMmod<0$ which would be non-physical. Thus, we have shown that 
\begin{equation}
	\CapG(\chPS,\Ninmod)=\CapGchi(\chPS,\Ninmod), \quad \Ninmod \ge \Nthr,
\end{equation}	
where $\CapGchi(\chPS,\Ninmod)$ is the right hand side of Eq.~\eqref{eq:CapGthrSup}. This proves the corollary.
\end{proof}
\subsection{Additional upper bounds on the classical capacity}
The upper bound on the classical capacity stated in Corollary~\ref{corCapGBounds} in the main text was obtained by generalizing the bounds that were found for thermal channels $\chPI$ (with $\tau>0$) in \cite{KS13}. Recently, additional upper bounds were obtained for the same channels \cite{GLMS12} and we extend them now to general channels $\chG$ with $\tau>0,y>0$. 

The bounds were obtained by maximizing the first term of the classical capacity [see its definition in Eqs.~\eqref{eq:capacity} and \eqref{eq:chi} in the main text] and by obtaining a lower bound $b$ on the second term, i.e. 
\begin{equation}
	\lim_{n \to \infty}\frac{1}{n}\min_{\mu}\int{\mu(dx) \, S\left((\chPI)^{\otimes n}[\rhoX]\right)} \ge b,
\end{equation}
where in total six bounds $b$ are presented in \cite{GLMS12}. We stated in the proof of Corollary~\ref{corCapGSup} that the fiducial channel is equivalent to a thermal channel preceded by an anti-squeezer and followed by a squeezer [see Fig.~\ref{fig:chBlochMessiah} (b)]. The following squeezer does not change the output entropy. Furthermore, one can always undo the preceding squeezer because the bound $b$ is not subject to an energy constraint. Therefore, any lower bound $b$ on the minimal output entropy of the thermal channel is as well a lower bound on the minimal output entropy of the fiducial channel. 

The first term of the classical capacity is known to be maximized by a thermal state carrying the total number of photons. Its entropy was already calculated in Eq.~\eqref{eq:CapMax}, i.e.
\begin{equation}\label{eq:CapMaxBound}
	\lim_{n \to \infty}\frac{1}{n}\max_{\mu \, : \, \rhoinmod \in \setN_{\Ninmod}} S\left((\chPS)^{\otimes n}[\rhoinmod]\right) = g\left(|\tau| \Ninmod + y\cosh(2s)+\frac{|\tau|-1}{2}\right).
\end{equation}
Corollary~\ref{corCap} in the main text states that $C(\chG,\Ninmod)=C(\chPS,\Ninmod)$. Therefore, any bound on the classical capacity of the fiducial channel $\chPS$ is also a bound on the classical capacity of a an arbitrary channel $\chG$. Thus, we obtained a list of upper bounds on the classical capacity that reads
\begin{equation}
	C(\chG,\Ninmod) \le g\left(|\tau| \Ninmod + y\cosh(2s)+\frac{|\tau|-1}{2}\right) - b, \quad \tau > 0, y>0,
\end{equation}
where $b$ has to be taken from \cite{GLMS12}. Note that with increasing $s$ those bounds become less and less tight because $b$ does not depend on $s$. However, as in the case of thermal channels \cite{GLMS12} some of those bounds in a certain region of channel parameters are tighter than the bound $\overline{C}$ given by Eq.~\eqref{eq:CapGBound} in the main text.
\end{widetext}


\begin{thebibliography}{9}
	
\bibitem{GaussRev1} 
C. M. Caves and P. D. Drummond, Rev.~Mod.~Phys. {\bf 66}, 481 (1994).

\bibitem{GaussRev2}
X.-B.~Wang, T.~Hiroshima, A.~Tomita, and M.~Hayashi, Phys.~Rep. \textbf{448}, 1 (2007).

\bibitem{G04} V.~Giovannetti, S.~Guha, S.~Lloyd, L.~Maccone, J.~H.~Shapiro, and H.~P.~Yuen, Phys.~Rev.~Lett. \textbf{92}, 027902 (2004).

\bibitem{Channels1} 
A.~S.~Holevo, M.~Sohma, and O.~Hirota, Phys.~Rev.~A {\bf 59}, 1820 (1999).

\bibitem{Channels2}
G.~Bowen, I.~Devetak, and S.~Mancini, Phys.~Rev.~A {\bf 71}, 034310 (2005).

\bibitem{Channels6}
S.~Guha, Phys.~Rev.~Lett. {\bf 106}, 240502 (2011).

\bibitem{Channels7}
R.~K\"onig and G.~Smith, Nature Photon. {\bf 7}, 142 (2013). 

\bibitem{ChannelsSH1} 
M.~Sohma and O.~Hirota, Phys.~Rev.~A \textbf{65}, 022319 (2002).

\bibitem{ChannelsSH2}
M.~Sohma and O.~Hirota, Phys.~Rev.~A \textbf{68}, 022303 (2003).

\bibitem{ChannelsSH3}
M.~Sohma and O.~Hirota, Phys.~Rev.~A \textbf{76}, 024303 (2007).


\bibitem{HW01} A.~S.~Holevo and R.~F.~Werner, Phys.~Rev.~A {\bf 63}, 032312 (2001).

\bibitem{H98} A.~S.~Holevo, Russ. Math. Surveys {\bf 53}, 1295 (1998). 

\bibitem{Channels3}
V.~Giovannetti and S.~Mancini, Phys.~Rev.~A {\bf 71}, 062304 (2005).

\bibitem{Channels4}
N.~J.~Cerf, J.~Clavareau, C.~Macchiavello, and J.~Roland, Phys.~Rev.~A {\bf 72}, 042330 (2005).

\bibitem{Channels5}
O.~V.~Pilyavets, V.~G.~Zborovskii, and S.~Mancini, Phys.~Rev.~A {\bf 77}, 052324 (2008).

\bibitem{H05} T.~Hiroshima, Phys.~Rev.~A {\bf 73}, 012330 (2006). 

\bibitem{LPM09} C.~Lupo,~O.~V.~Pilyavets, and S.~Mancini, New~J.~Phys. {\bf 11}, 063023 (2009).

\bibitem{SDKC09} J.~Sch\"afer, D.~Daems, E.~Karpov, and N.~J.~Cerf, Phys.~Rev.~A {\bf 80}, 062313 (2009).

\bibitem{PLM09} O.~V.~Pilyavets, C.~Lupo, and S.~Mancini, IEEE Trans. Inf. Theory {\bf 58}, 6126 (2012). 

\bibitem{SKC10} J.~Sch\"afer, E.~Karpov, and N.~J.~Cerf, in \textit{Proceedings of SPIE} (SPIE-The International Society for Optical Engineering, Bellingham, WA, 2010), vol.~7727, p.~77270J.

\bibitem{SKC11} J.~Sch\"afer, E.~Karpov, and N.~J.~Cerf, Phys.~Rev.~A {\bf 84}, 032318 (2011).

\bibitem{SKC12} J.~Sch\"afer, E.~Karpov, and N.~J.~Cerf, Phys.~Rev.~A {\bf 85}, 012322 (2012).

\bibitem{GaussianComment} In the rest of the Letter, we refer to the classical capacity simply as the capacity. Similarly, we use the term Gaussian capacity to denote the 
Gaussian classical capacity. The latter corresponds to the classical capacity obtained with a restriction to Gaussian individual symbol states and Gaussian averaged states.

\bibitem{KS13} R.~K\"onig and G.~Smith, Phys.~Rev.~Lett. {\bf 110}, 040501 (2013).

\bibitem{GGLMS04} V.~Giovannetti, S.~Guha, S.~Lloyd, L.~Maccone, and J.~H.~Shapiro, Phys.~Rev.~A {\bf 70}, 032315 (2004).

\bibitem{DeGosson} M.~De Gosson, \emph{Symplectic Geometry and Quantum Mechanics: Operator Theory, Advances and Applications} (Birkh\"auser, Basel, 2006), Vol.~166.

\bibitem{GNLSC12} R. Garc\'ia-Patr\'on, C.~Navarrete-Benlloch, S.~Lloyd, J.~H.~Shapiro, and N.~J.~Cerf, Phys.~Rev.~Lett. \textbf{108}, 110505 (2012).

\bibitem{H08} A.~S.~Holevo, Probl. Inf. Trans. {\bf 44}, 3 (2008). 

\bibitem{H07CGH061} 
F.~Caruso, V.~Giovannetti, and A.~S.~Holevo, New~J.~Phys. {\bf 8}, 310 (2006).

\bibitem{H07CGH062}
A.~S.~Holevo, Probl. Inf. Trans. {\bf 43}, 1 (2007).

\bibitem{H07CGH063}
J.~S.~Ivan, K.~K.~Sabapathy, and R.~Simon, Phys.~Rev.~A \textbf{84}, 042311 (2011). 

\bibitem{SUPPMAT} See Supplemental Material for more details on Table~\ref{table:channels} as well as on the derivation of Theorem 1 and its corollaries.

\bibitem{SW97} B.~Schumacher and M.~D.~Westmoreland, Phys.~Rev.~A {\bf 56}, 131 (1997). 

\bibitem{H09} M.~B.~Hastings, Nature~Phys. {\bf 5}, 255 (2009). 

\bibitem{S02HS041} 
P.~W.~Shor, J.~Math. Phys. {\bf 43}, 4334 (2002).

\bibitem{S02HS042}
A.~S.~Holevo and M.~E.~Shirokov, Comm.~Math.~Phys. {\bf 249}, 417 (2004).

\bibitem{EW07} J.~Eisert and M.~M.~Wolf, in \emph{Quantum Information with Continuous Variables of Atoms and Light}, edited by N.~J.~Cerf, G.~Leuchs and E.~S.~Polzik (Imperial College Press, London, 2007), pp. 23-42.

\bibitem{GC2013} R. Garc\'ia-Patr\'on and N.~J.~Cerf, (to be published).

\bibitem{GLMS12} V.~Giovannetti, S.~Lloyd, L.~Maccone, and J.~H.~Shapiro, arXiv:1210.3300v1. 

\bibitem{FCP041} 
S. L. Braunstein, Phys. Rev. Lett. \textbf{80}, 4084 (1998).

\bibitem{FCP042}
J.~Fiurasek, N.~J.~Cerf, and E.~S.~Polzik, Phys.~Rev.~Lett. {\bf 93}, 180501 (2004).

\bibitem{WGKWC04} M.~M.~Wolf, G.~Giedke, O.~Kr\"uger, R.~F.~Werner, and J.~I.~Cirac, Phys. Rev. A {\bf 69}, 052320 (2004).

\bibitem{B05} S.~L.~Braunstein, Phys.~Rev.~A {\bf 71}, 055801 (2005).

\end{thebibliography}
\end{document}